\documentclass{LMCS}

\def\doi{8(4:16)2012}
\lmcsheading%
{\doi}
{1--38}
{}
{}
{Apr.~10, 2011}
{Nov.~21, 2012}
{}   

\newcommand{\cexcex}[1]{}%

\def\qedhere{}
\let\origparagraph\paragraph
\renewcommand{\paragraph}[1]{\origparagraph{\textbf{#1}}}

\newenvironment{definition}[1][]
{\begin{defi}\ifthenelse{\equal{#1}{}}{}{(\textsc{#1}). }\ignorespaces
}%
{\end{defi}\par\noindent\ignorespacesafterend}

\newenvironment{theorem}[1][]
{\begin{thm}\ifthenelse{\equal{#1}{}}{}{(\textsc{#1}). }\ignorespaces
}%
{\end{thm}\par\noindent\ignorespacesafterend}

\newenvironment{lemma}[1][]
{\begin{lem}\ifthenelse{\equal{#1}{}}{}{(\textsc{#1}). }\ignorespaces
}%
{\end{lem}\par\noindent\ignorespacesafterend}

\newenvironment{proposition}[1][]
{\begin{prop}\ifthenelse{\equal{#1}{}}{}{(\textsc{#1}). }\ignorespaces
}%
{\end{prop}\par\noindent\ignorespacesafterend}

\newenvironment{corollary}[1][]
{\begin{cor}\ifthenelse{\equal{#1}{}}{}{(\textsc{#1}). }\ignorespaces
}%
{\end{cor}\par\noindent\ignorespacesafterend}

\newenvironment{example}[1][]
{\begin{exa}\ifthenelse{\equal{#1}{}}{}{(\textsc{#1}). }\ignorespaces
}%
{\end{exa}\par\noindent\ignorespacesafterend}

\usepackage{array}
\usepackage{fixltx2e}
\usepackage{cite}

\usepackage{ifthen}
\usepackage{tikz}
\usetikzlibrary{automata,trees,shapes,arrows,decorations,decorations.pathmorphing}
\usepackage{wrapfig}
\usepackage{wasysym}

\usepackage{prettyref}
\newcommand{\rref}[2][]{\prettyref{#2}}
\newrefformat{sec}{Sect.~\ref{#1}}
\newrefformat{app}{Appendix~\ref{#1}}
\newrefformat{def}{Def.~\ref{#1}}
\newrefformat{thm}{Theorem~\ref{#1}}
\newrefformat{prop}{Proposition~\ref{#1}}
\newrefformat{lem}{Lemma~\ref{#1}}
\newrefformat{cor}{Corollary~\ref{#1}}
\newrefformat{ex}{Example~\ref{#1}}
\newrefformat{tab}{Table~\ref{#1}}
\newrefformat{fig}{\figurename~\ref{#1}}
\newrefformat{case}{case~\ref{#1}}

\usepackage{enumerate,hyperref}

\usepackage[bbsets]{msymb}
\usepackage[Dfprime]{manalysis}
\usepackage[modernsign,substopindex,longmquant,mquantifiertype,mconnectiveformal,bracketinterpret,prefixinterpret,modifopindex,seqarrow,seqoptional,footnotecalculus,abbrseqcontext,shortterms,nosigmaterms,novarterms]{logic}
\usepackage[pretest,nocommandblocks]{progreg}
\usepackage[bracketinterpret,prefixinterpret]{dL}

\newcommand{\bebecomes}{\mathrel{::=}}

\newcommand{\alternative}{~|~}

\newcommand{\DIclass}{\ensuremath{\mathcal{DI}}}%
\newcommand{\DCIclass}{\ensuremath{\mathcal{DCI}}}%
\newcommand{\DI}[1][]{\ensuremath{\DIclass_{\ifthenelse{\equal{#1}{}}{}{#1}}}\xspace}
\newcommand{\DIp}[1][]{\ensuremath{\DIclass_{\ifthenelse{\equal{#1}{}}{\land,\lor}{#1,\land,\lor}}}\xspace}
\newcommand{\DIS}[1][]{\ensuremath{\DCIclass_{\ifthenelse{\equal{#1}{}}{}{#1}}}\xspace}
\newcommand{\DISp}[1][]{\ensuremath{\DCIclass_{\ifthenelse{\equal{#1}{}}{\land,\lor}{#1,\land,\lor}}}\xspace}

\newcommand{\DIop}[1][]{\ensuremath{\DIclass^\circ_{\ifthenelse{\equal{#1}{}}{}{#1}}}\xspace}
\newcommand{\DIopp}[1][]{\ensuremath{\DIclass^\circ_{\ifthenelse{\equal{#1}{}}{\land,\lor}{#1,\land,\lor}}}\xspace}
\newcommand{\DISop}[1][]{\ensuremath{\DCIclass^\circ_{\ifthenelse{\equal{#1}{}}{}{#1}}}\xspace}
\newcommand{\DISopp}[1][]{\ensuremath{\DCIclass^\circ_{\ifthenelse{\equal{#1}{}}{\land,\lor}{#1,\land,\lor}}}\xspace}

\newcommand{\ivr}{H}
\newcommand{\inv}{F}
\newcommand{\precond}{A}

\renewcommand{\linferenceRuleNameSeparation}{\hspace{5pt}}

\newcommand{\stdI}{\dLint[state=\nu]}
\newcommand{\I}{\stdI}

\usepackage{ifpdf}
\ifpdf
\fi

\begin{document}
\title[The Structure of Differential Invariants and Differential Cut Elimination]{The Structure of Differential Invariants and Differential Cut Elimination}
\author[A.~Platzer]{Andr\'e Platzer}
\address{Carnegie Mellon University, Computer Science Department, Pittsburgh, PA, USA}
\email{aplatzer@cs.cmu.edu}
\thanks{
This material is based upon work supported by the National Science Foundation under
NSF CAREER Award CNS-1054246, NSF EXPEDITION CNS-0926181, and under Grant No.
CNS-0931985.
}

\begin{abstract}
  The biggest challenge in hybrid systems verification is the handling of differential equations. Because computable closed-form solutions only exist for very simple differential equations, proof certificates have been proposed for more scalable verification.
  Search procedures for these proof certificates are still rather ad-hoc, though, because the problem structure is only understood poorly.
  We investigate \emph{differential invariants}, which define an induction principle for differential equations and which can be checked for invariance along a differential equation just by using their differential structure, without having to solve them.
  We study the structural properties of differential invariants.
  To analyze trade-offs for proof search complexity, we identify more than a dozen relations between several classes of differential invariants and compare their deductive power.
  As our main results, we analyze the deductive power of differential cuts and the deductive power of differential invariants with auxiliary differential variables.
  We refute the \emph{differential cut elimination hypothesis} and show that, unlike standard cuts, differential cuts are fundamental proof principles that strictly increase the deductive power.
  We also prove that the deductive power increases further when adding \emph{auxiliary differential variables} to the dynamics.

\end{abstract}

\maketitle

\newsavebox{\Rval}%
\sbox{\Rval}{$\mathbb{R}$}

\section{Introduction}

Hybrid systems \cite{Tavernini87,DBLP:conf/lics/Henzinger96,DBLP:journals/tac/BranickyBM98,DavorenNerode_2000} are systems with joint discrete and continuous dynamics, e.g., aircraft that move continuously in space along differential equations for flight and that are controlled by discrete control decisions for flight control like collision avoidance maneuvers.
Hybrid systems verification is an important but challenging and undecidable problem \cite{DBLP:conf/lics/Henzinger96,DBLP:journals/tac/BranickyBM98}.
Several verification approaches for hybrid systems have been proposed.
Verifying properties of differential equations is at the heart of hybrid systems verification.
In fact, we can prove properties of hybrid systems exactly as good as properties of differential equations can be proved.
This surprising intuition is made formally rigorous by our relatively complete axiomatization of our logic for hybrid systems relative to properties of differential equations \cite{DBLP:journals/jar/Platzer08}.
One direction is obvious, because differential equations are part of hybrid systems, so we can only understand hybrid systems to the extent that we can reason about their differential equations.
We have proven the other direction as well, by proving that all true properties of hybrid systems can be reduced effectively to elementary properties of differential equations \cite{DBLP:journals/jar/Platzer08}.
Moreover, we gave a proof calculus for hybrid systems that performs this reduction constructively and, vice versa, provides a provably perfect lifting of every approach for differential equations to hybrid systems \cite{DBLP:journals/jar/Platzer08}.
The proof of this result is constructive, but also very complicated, which explains why some approaches for differential equations are easier to lift to hybrid systems in practice than others.
In this paper, we consider an approach (induction for differential equations) that lifts naturally.

Thus, the remaining (yet undecidable) question is how to prove properties of differential equations.
If the differential equation has a simple polynomial solution, then this is easy \cite{DBLP:journals/jar/Platzer08} using the decidable theory of first-order real arithmetic \cite{tarski_decisionalgebra51}.
Unfortunately, almost no differential equations have such simple solutions.
Polynomial solutions arise in linear differential equations with constant coefficients where the coefficient matrix is nilpotent. But this is a very restricted class.
For other differential equations, numerous approximation techniques have been considered to obtain approximate answers \cite{DBLP:conf/hybrid/GreenstreetM99,DBLP:conf/hybrid/AsarinDG03,DBLP:journals/tac/GirardP07,RatschanS07,DBLP:journals/sttt/Frehse08}.
It is generally surprisingly difficult to get them formally sound, however, due to inherent numerical approximation that make the numerical image computation problem itself neither semidecidable nor co-semidecidable \cite{DBLP:conf/hybrid/PlatzerC07,DBLP:journals/mst/Collins07}, even when tolerating arbitrarily large error bounds on the decision.

As alternative approaches that are not based on approximation, proof certificate techniques have been proposed for hybrid systems verification, including barrier certificates \cite{DBLP:conf/hybrid/PrajnaJ04,DBLP:journals/tac/PrajnaJP07}, equational templates \cite{DBLP:journals/fmsd/SankaranarayananSM08}, differential invariants \cite{DBLP:journals/logcom/Platzer10,DBLP:conf/cav/PlatzerC08,DBLP:journals/fmsd/PlatzerC09}, and a constraint-based template approach \cite{DBLP:conf/cav/GulwaniT08}.
Once a proof certificate has been found, it can be checked efficiently.
But we first have to find it.
Previous search procedures are based on searching for parameters of various user-specified templates \cite{DBLP:conf/hybrid/PrajnaJ04,DBLP:journals/tac/PrajnaJP07,DBLP:journals/fmsd/SankaranarayananSM08,DBLP:conf/cav/GulwaniT08,DBLP:conf/cav/PlatzerC08,DBLP:journals/fmsd/PlatzerC09}.
But these verification techniques fail if the template does not include the required form.
How do we need to choose the templates?
What are the trade-offs for choosing them?
If we choose a smaller class of witness templates to search through, could this make the search more difficult? Or would we miss out on proofs altogether? Would we no longer be able to prove some properties at all even though they are true?
The primary reason why search procedures for hybrid systems proof certificates are still ad-hoc is that the structure of proof certificates of hybrid systems has not been understood very well so far.

In this paper, we identify and characterize the structure of hybrid systems proof certificates.
What relationships exist between various choices for classes of proof certificates?
Are there system properties that cannot be proven when focusing on a particular class of invariants?
Are any of the choices superior to others or are they mutually incomparable?
Invariants are well-understood for discrete systems but not yet for continuous and hybrid systems.

\begin{figure}[tb]
  \centering
  \newcommand{\reason}[2][]{\textcolor{blue}{#1}}
  \newcommand{\incomparable}{i}
  \begin{tikzpicture}[xscale=1.4,yscale=1.4]%
    \tikzstyle{cross line}=[preaction={draw=white,-,line width=8pt}]
    \tikzstyle{back line}=[]%
    \tikzstyle{less}=[arrows=right hook-stealth']
    \tikzstyle{lesseq}=[arrows=right to reversed-stealth']
    \tikzstyle{eq}=[double distance=2pt]
    \tikzstyle{suspected}=[gray,dashed]
    \tikzstyle{unknown}=[gray,dotted]
    \tikzstyle{incomp}=[decorate,decoration=zigzag]
    \node (eq) at (1,0) {$\DI[=]$};
    \node (peq) at (2.5,0) {$\DIp[=]$};
    \node (gt) at (0,-1) {$\DI[>]$};
    \node (pgt) at (3.5,-1) {$\DIp[>]$};
    \node (geq) at (0,1) {$\DI[\geq]$};
    \node (pgeq) at (3.5,1) {$\DIp[\geq]$};
    \node (di) at (7,0) {$\DI$}; %
    \node (pgeqeq) at (5.5,1) {$\DIp[\geq,=]$};
    \node (pgteq) at (5.5,-1) {$\DIp[>,=]$};
    \draw[incomp] (eq) --node[sloped,below] {\reason[\ref{prop:eq},\ref{prop:gt}]{\incomparable}} (gt); 
    \draw[less] (eq) --node[sloped,above] {\reason[\ref{prop:eq},\ref{prop:eq-geq}]{$\leq$}} (geq);
    \draw[incomp] (gt) --node[left] {\reason[\ref{prop:atom-p}]{\incomparable}} (geq);
    \draw[eq] (eq) -- node[above]{\reason[\ref{prop:EDP}]{$=$}} (peq);
    \draw[less] (geq) -- node[above]{\reason[\ref{thm:atomic}]{$<$}} (pgeq);
    \draw[less] (gt) -- node[below]{\reason[\ref{thm:atomic}]{$<$}} (pgt);
    \draw[less] (pgeqeq) -- node[above,sloped]{\reason[\ref{cor:p-incomp}]{$<$}} (di);
    \draw[less] (pgt) -- node[below]{\reason[\ref{prop:atom-p}]{$<$}} (pgteq); %
    \draw[less] (peq) -- node[sloped,above]{\reason[\ref{prop:eq-geq}]{$<$}} (pgeq);
    \draw[eq] (pgeq) -- node[above]{\reason[\ref{prop:eq-geq}]{$=$}} (pgeqeq);
    \draw[incomp,back line,cross line] (pgt) -- node[right,pos=0.2] {\reason[\ref{prop:atom-p}]{\incomparable}} (pgeq);
    \draw[incomp] (peq) --node[below,sloped] {\reason[\ref{cor:p-incomp}]{\incomparable}} (pgt);
    \draw[incomp,back line] (pgteq) -- node[left,pos=0.25] {\reason[\ref{prop:atom-p}]{\incomparable}}
    node[right,pos=0.25] {\reason[\ref{thm:pgteq-di}]{\incomparable}} (pgeqeq);
    \draw[less] (pgteq) -- node[below,sloped]{\reason[\ref{thm:pgteq-di}]{$<$}} (di);
    \draw[less,cross line] (peq) -- node[above,pos=0.42]{\reason[\ref{prop:eq}]{$<$}} %
    (di);
    \draw[less,back line,shorten <=2pt] (peq) to[bend left=10] node[above,sloped,pos=0.6]{\reason[\ref{cor:p-incomp}]{$<$}} 
      (pgteq); %
      \begin{scope}[xshift=8.4cm,yshift=0cm]
      \draw[less,shorten >=6pt,shorten <=6pt] (0,1) node {$\mathcal{A}$\,\,\,}--node[above] {$\mathcal{A}<\mathcal{B}$} node[below,text width=1cm] {strict\\inclusion} ++(1.3,0) node {~$\mathcal{B}$};
      \draw[eq,shorten >=6pt,shorten <=6pt] (0,-0.2) node {$\mathcal{A}$\,\,\,}--node[above] {$\mathcal{A}\mequiv \mathcal{B}$} node[below=4pt] {equivalent} ++(1.3,0) node {~$\mathcal{B}$};
      \draw[incomp,shorten >=4pt,shorten <=4pt] (0,-1) node {$\mathcal{A}$\,\,\,}--node[below=4pt] {incomparable} ++(1.3,0) node {~$\mathcal{B}$};
      \end{scope}
  \end{tikzpicture}
  \caption{Differential invariance chart (proposition numbers are indicated for each relation)}
  \label{fig:diffind-chart}
\end{figure}
We consider \emph{differential invariants}, which include several previous approaches as special cases (yet in modified forms to make the reasoning sound).
Logical proofs with differential invariants have been instrumental in enabling the verification of several practical hybrid systems, including separation properties in complex curved flight collision avoidance maneuvers for air traffic control \cite{DBLP:conf/fm/PlatzerC09}, advanced safety, reactivity and controllability properties of train control systems with disturbance and PI controllers \cite{DBLP:conf/icfem/PlatzerQ09}, and properties of electrical circuits \cite{Platzer10}.
Our logic-based proof approach with differential invariants has been the key enabling technique to make formal verification of these systems possible.

Here, we study the structure of differential invariants from a more foundational perspective and develop their proof theory.
Differential equations enjoy various universal computation properties, hence verification is not even semidecidable \cite{DBLP:journals/tcs/Branicky95,GracaCB07,BournezCGH07,DBLP:journals/mst/Collins07}.
Consequently, every complete proof rule is unsound or ineffective.
Hence, proof theory is not a study of completeness but a study of alignment and relative provability.

We analyze the relationships between several classes of differential invariants.
Our main tool for this is the study of \emph{relative deductive power}.
For comparing two classes of differential invariants, $\mathcal{A}$ and $\mathcal{B}$, we investigate whether there is a system property that only $\mathcal{A}$ can prove or whether all properties that can be proven using $\mathcal{A}$ can also be proven using $\mathcal{B}$.
If the answer is yes (inclusion), we give a construction that translates proofs using $\mathcal{A}$ into proofs using $\mathcal{B}$.
If the answer is no (separation), we prove for a formal proof using $\mathcal{A}$ that there is no formal proof using $\mathcal{B}$.
Of course, there are infinitely many possible proofs to check.
We, thus, show the deductive power separation properties by coming up with an indirect characteristic that separates the properties provable using $\mathcal{B}$ from the particular proof using $\mathcal{A}$.
These separation properties that we identify in our proofs are of more general interest beyond the systems we show.
We identify more than a dozen (16) relationships between nine classes of differential invariants (summarized in \rref{fig:diffind-chart}),
which shed light on how the classes compare in terms of their deductive power for systems analysis.

While our study is mostly one of logically fundamental properties like deductive power and provability, our proofs indicate additional computational implications, e.g., to what extent the polynomial degree or formula complexity increases with the respective inclusion and equivalence reductions shown in \rref{fig:diffind-chart}.
Our results summarized in \rref{fig:diffind-chart} show, for example, that equational differential invariants do not need Boolean operators, yet the required polynomial degree increases.
It is also easy to read off general limits of classes of differential invariants from our proofs of the separation properties, which can be exploited constructively for differential invariant search.

Observe that the algebraic structure of nonlinear real arithmetic alone is not sufficient to explain the relations identified in \rref{fig:diffind-chart}.
Both the algebraic and the differential structure of differential invariants matter for the answer, because the dynamics along differential equations determines which properties hold when following the system dynamics.
Consequently, the differential structure of the differential equation and of the property matter.
Even if the real algebraic structures match, we still do not know if the corresponding differential structures align in a compatible way.
We will see that the joint differential-algebraic structure of the problem can be surprising even for very simple differential equations already. In particular, our observations are fundamental and cannot be sidestepped by restricting attention to other classes of differential equations.

Ultimately, invariance crucially depends on the differential-geometrical structure induced by the differential equation.
We, thus, study the proof-theoretical properties of what can be proved about a differential equation based on its differential and algebraic invariant structure.
Most importantly, and most surprisingly, we refute the \emph{differential cut elimination hypothesis} \cite{DBLP:journals/logcom/Platzer10}.
Differential cuts \cite{DBLP:journals/logcom/Platzer10} have a simple intuition.
Similar to a cut in first-order logic, they can be used to first prove a lemma and then use it.
By the seminal cut elimination theorem of Gentzen \cite{Gentzen35I,Gentzen35}, standard logical cuts in first-order logic do not change the deductive power and can always be eliminated from proofs.
Unlike standard cuts, however, differential cuts work for differential equations and can be used to change the dynamics of the system.
Does this differential cut proof principle also support differential cut elimination?
Are differential cuts only a convenient proof shortcut?
Or are differential cuts an independent fundamental proof principle?
We show that the addition of differential cuts actually increases the deductive power.
There are properties of differential equations that can only be proven using differential cuts, not without them.
Hence, differential cuts indeed turn out to be a fundamental proof principle.
Four years ago we had conjectured that differential cuts are necessary to prove a certain class of air traffic control properties \cite{DBLP:journals/logcom/Platzer10}.
We have now refuted this conjecture, since those differential cuts can still be eliminated with a clever construction.
But we show that differential cuts are still necessary in general.
This illustrates the surprisingly subtle nature of proving properties of differential equations.

Furthermore, we present new proof rules for auxiliary differential variables and prove that the addition of auxiliary differential variables increases the deductive power, even in the presence of differential cuts.
That is, there are system properties that can only be proven using auxiliary differential variables in the dynamics relative to which the original system variables relate.
Hence, auxiliary differential variables are also a fundamental proof principle.
This is similar to discrete programs where auxiliary variables may also be necessary to prove some properties.
We now show that the same also holds for differential equations even in the presence of differential cuts. Refining differential equations with auxiliary differential variables adds to the deductive power, which, surprisingly, has not been considered before.

These alignments of the relative deductive power shed light on the properties and practical implications of various choices for hybrid systems verification.
They help make informed decisions concerning which restrictions on proof search (and differential invariant search) are tolerable without changing the deductive power.
In this paper we study the problem of proving properties of differential equations.
This directly relates to a study of proving properties of hybrid systems by way of our proof calculus from previous work that we have shown to be a complete axiomatization of hybrid systems relative to properties of differential equations \cite{DBLP:journals/jar/Platzer08}.
This previous result makes it formally precise how the verification of hybrid systems can be reduced directly to the verification of properties of differential equations, and how our new results about differential equations lift to hybrid systems \cite{DBLP:journals/jar/Platzer08,DBLP:journals/logcom/Platzer10,Platzer10}.

Our research requires a symbiosis of logic with elements of real arithmetical, differential, semialgebraic, and geometrical principles. 
Based on the results presented in this paper, we envision a continuing development of a new field that we call \emph{real differential semialgebraic geometry}.
In this work, it is of paramount importance to distinguish semantical truth from deductive proof.
It does not help if a property is true, unless we can produce a proof to show that it is true.
We assume that the reader is familiar with the proof theory of classical logic \cite{Fitting96a,Andrews02}, including the underlying notions of formal deduction, and the relationship and differences between syntax, semantics, and proof calculi.
We also assume basic knowledge of differential equations \cite{Walter:ODE} and of first-order real arithmetic \cite{tarski_decisionalgebra51}.

\section{Preliminaries}
Continuous dynamics described by differential equations are a crucial part of hybrid system models.
An important subproblem in hybrid system verification is the question whether a system following a (vectorial) differential equation \m{\D{x}=\theta} that is restricted to an \emph{evolution domain constraint} region $\ivr$ will always stay in the region $F$.
We represent this by the modal formula \m{\dbox{\hevolvein{\D{x}=\theta}{\ivr}}{\inv}}.
It is true at a state~$\iget[state]{\I}$ if, indeed, a system following \m{\D{x}=\theta} from~$\iget[state]{\I}$ will always stay in $\inv$ at all times (at least as long as the system stays in $\ivr$).
It is false at~$\iget[state]{\I}$ if the system can follow \m{\D{x}=\theta} from~$\iget[state]{\I}$ and leave $\inv$ at some point in time, without having left $\ivr$ at any time.
Here, $\inv$ and $\ivr$ are (quantifier-free) formulas of real arithmetic and \m{\D{x}=\theta} is a (vectorial) differential equation, i.e., \m{x=(x_1,\dots,x_n)} is a vector of variables and \m{\theta=(\theta_1,\dots,\theta_n)} a vector of polynomial terms;
for extensions to rational functions, see \cite{DBLP:journals/logcom/Platzer10}.
In particular, $\ivr$ describes a region that the continuous system cannot leave (e.g., because of physical restrictions of the system or because the controller otherwise switches to another mode of the hybrid system).
In contrast, $\inv$ describes a region for which we want to prove that the continuous system \m{\hevolvein{\D{x}=\theta}{\ivr}} will never leave it.
If, for instance, the formula \m{\ivr\limply\inv} is valid (i.e., the evolution domain region $\ivr$ is contained in region $\inv$), then \m{\dbox{\hevolvein{\D{x}=\theta}{\ivr}}{\inv}} is valid trivially.
This reasoning alone rarely helps, because $\ivr$ will not be contained in $\inv$ in the interesting cases.
For example, in an aircraft controller, the evolution domain constraint $\ivr$ could be that the air speed is nonnegative, whereas the property $\inv$ that we want to prove is that the aircraft do not collide, which is in no way entailed by $\ivr$.

\paragraph{Differential Dynamic Logic (Excerpt)}
The modal logical principle described above can be extended to a full dynamic logic for hybrid systems, called differential dynamic logic \dL \cite{DBLP:journals/jar/Platzer08,DBLP:journals/logcom/Platzer10,Platzer10}.
Here we only need propositional operators and modalities for differential equations.
For our purposes, it is sufficient to consider the \dL fragment with the following grammar (where $x$ is a vector of variables, $\theta$ a vector of terms of the same dimension, and $F,\ivr$ are formulas of (quantifier-free) first-order real arithmetic over the variables $x$):
\[
\phi,\psi ~\bebecomes~ 
F
\alternative \phi \land \psi
\alternative \phi \lor \psi
\alternative \phi \limply \psi
\alternative \phi \lbisubjunct \psi
\alternative \dbox{\hevolvein{\D{x}=\theta}{\ivr}}{F}
\]
A state is a function \m{\iget[state]{\I}:V\to\reals} that assigns real numbers to all variables in the set \m{V=\{x_1,\dots,x_n\}}.
We denote the value of term $\theta$ in state $\iget[state]{\I}$ by \m{\ivaluation{\I}{\theta}}.
The semantics is that of first-order real arithmetic with the following addition:
\begin{iteMize}{$\bullet$}
\item \m{\imodels{\I}{\dbox{\hevolvein{\D{x}=\theta}{\ivr}}{F}}} iff
for each function \m{\varphi:[0,r]\to(V\to\reals)} of some duration $r$       we have \m{\varphi(r) \models F} under the following two conditions:
\newcommand{\Ifz}{\iconcat[state=\varphi(\zeta)]{\stdI}}%
      \begin{enumerate}[(1)]
        \item the differential equation holds, i.e., for each variable~$x_i$ and each time $\zeta \in [0,r]$:
      \[
      \D[t]{\,\ivaluation{\iconcat[state=\varphi(t)]{\stdI}}{x_i}} (\zeta) = \ivaluation{\iconcat[state=\varphi(\zeta)]{\stdI}}{\theta_i}
      \]
      in particular, \m{\ivaluation{\iconcat[state=\varphi(t)]{\stdI}}{x_i}} has to be differentiable at $\zeta$ as a function of $t$
      \item and the evolution domain is always respected, i.e.,
      \m{\imodels{\Ifz}{\ivr}} for each $\zeta \in [0,r]$.
      \end{enumerate}
\end{iteMize}
Other details about the logic, its semantics, and proof rules that are not of immediate concern here can be found in previous work \cite{DBLP:journals/jar/Platzer08,DBLP:journals/logcom/Platzer10,Platzer10}.
We do not need to consider the full logic and full proof calculus here, because both are strictly compositional.
The other proof rules deal with handling other features of hybrid systems like their discrete dynamics, sequential compositions, nondeterministic choices, and repetitions.
For our purposes, it is sufficient to assume a decision procedure for first-order logic of real-closed fields \cite{tarski_decisionalgebra51} and a propositionally complete base calculus.
For simplicity, we also allow standard logical cuts just to have a simple way of glueing multiple proofs together.
In the sequel, we denote the use of instances of valid tautologies of first-order real arithmetic in proofs by \irref{qear}.
For reference, these background proof rules are summarized in \rref{app:proof-rules}.

\paragraph{Solutions as Explicit Witnesses}
{\let\precond\inv%
An explicit witness for the validity of a formula like \m{\precond \limply \dbox{\hevolvein{\D{x}=\theta}{\ivr}}{\inv}} would be the set of all solutions of the differential equation along with a proof that, when starting in any state that satisfies $\precond$, formula $\inv$ holds all along every solution of \m{\D{x}=\theta} as long as formula $\ivr$ holds.
If we happen to know a (unique) solution \m{X(t)=f(t,x_0)} of the differential equation \m{\D{x}=\theta} with a function $f(t,x_0)$ of time $t$ and the initial state $x_0$, then we have the following sound rule

\[
  {\linferenceRule[sequent]
    {\lsequent{\precond~}{~\forall{r}{\left(r\geq0\land\forall{\zeta}{\left(0\leq\zeta\leq r \limply \subst[\ivr]{x}{f(\zeta,x)}\right)}\limply \subst[\inv]{x}{f(r,x)}\right)}}}
    {\lsequent{\precond}{\dbox{\hevolvein{\D{x}=\theta}{\ivr}}{\inv}}}
}{}
\]
where $\subst[\inv]{x}{f(r,x)}$ is the result of applying to $\inv$ the substitution that replaces the variable $x$ by $f(r,x)$ and similarly for $\subst[\ivr]{x}{f(\zeta,x)}$ .
It is very easy to see why this rule is sound \cite{DBLP:journals/jar/Platzer08}, because it directly follows the semantics.
The problem is that it is usually not a useful proof rule, because it is rarely effective.
It only helps if we can effectively compute a (unique) solution $f(t,x_0)$, as a function of $t$ and $x_0$, to the symbolic initial-value problem \m{\D{x}=\theta,x(0)=x_0} for a variable symbol $x_0$.
Notice that conventional initial-value problems are numerical with concrete numbers $x_0\in\reals^n$ as initial values, not symbols \cite{Walter:ODE}.
This is not enough for our purpose, because we need to prove that the formula holds for all states satisfying initial assumption $\precond$, which could be uncountably many.
We can hardly solve uncountably many different initial-value problems to verify a system.
Also, the rule only helps when the resulting arithmetic is computable and the formula with the (alternating) quantifiers $\forall r$ and $\forall \zeta$ in the premise can be decided.
Even very simple linear differential equations like \m{\hevolve{\D{x}=y\syssep\D{y}=-x}} have trigonometric functions as solutions, which leads to undecidable arithmetic by a simple corollary \cite[Theorem 2]{DBLP:journals/jar/Platzer08} to G\"odel's incompleteness theorem \cite{Goedel_1931}.
For most differential equations, the solutions cannot be computed effectively, fall outside decidable classes of arithmetic, or do not even exist in closed form.

Semantical approaches to proving properties of differential equations are not very informative for actual provability.
We need to consider the problem from a proof-theoretic perspective and investigate syntactic proof rules that are computationally effective, i.e., they lead to computable or decidable formulas and have computable side conditions.
Coming up with computationally ineffective proof rules for differential equations would obviously be trivial, even if they are sound and complete.
The appropriate fundamental question to ask is how provability compares and aligns for different choices of sound and effective proof rules.
This is what we address in this paper.
}

\section{Differential Invariants \& Differential Cuts} \label{sec:diffind}
The most fundamental question about a differential equation for safety verification purposes is whether a formula $\inv$ is an invariant, i.e., whether formula \m{\inv \limply \dbox{\hevolvein{\D{x}=\theta}{\ivr}}{\inv}} is valid (true in all states).
At first sight, invariance questions may look like a somewhat special case (pre- and postcondition are the same $\inv$ here), but they are really at the heart of the hybrid systems verification problem.
All more complicated safety properties of hybrid systems reduce to a series of invariance questions using the proof calculus that we presented in previous work \cite{DBLP:journals/jar/Platzer08,Platzer10}.
For instance, formulas of the form \m{\precond \limply \dbox{\hevolvein{\D{x}=\theta}{\ivr}}{B}} can be derived using the usual variation
\begin{equation}
\linfer
{\lsequent{\precond}{\inv}
& \lsequent{\inv}{\dbox{\hevolvein{\D{x}=\theta}{\ivr}}{\inv}}
& \lsequent{\inv}{B}}
{\lsequent{\precond}{\dbox{\hevolvein{\D{x}=\theta}{\ivr}}{B}}}
  \label{eq:odeb-indirect}
\end{equation}
We will use this variation occasionally.
Formally, variation~\rref{eq:odeb-indirect} can be derived in proof calculi using standard propositional cuts and the G\"odel generalization:

\centerline{
\begin{calculus}
      \cinferenceRule[genb|{$[]gen$}]{$\ddiamond{}{}/\dbox{}{}$ generalization}
      {\linferenceRule[sequent]
        {\lsequent{F}{G}}
        {\lsequent{\dbox{\hevolvein{\D{x}=\theta}{\ivr}}{F}}{\dbox{\hevolvein{\D{x}=\theta}{\ivr}}{G}}}
      }{}
\end{calculus}
}
What we need to do to use \rref{eq:odeb-indirect} effectively, however, is to find a good choice for the invariant $\inv$ that makes \m{\inv \limply \dbox{\hevolvein{\D{x}=\theta}{\ivr}}{\inv}} valid.
For this, we need to understand which formulas are good candidates for invariants.
\begin{definition}[Invariant]
  Formula $\inv$ is called an \emph{invariant} of the system \m{\hevolvein{\D{x}=\theta}{\ivr}} if the formula \m{\inv \limply \dbox{\hevolvein{\D{x}=\theta}{\ivr}}{\inv}}  is valid.
\end{definition}
Invariance is defined in terms of validity of differential dynamic logic formulas, which is a semantic concept and neither decidable nor semidecidable \cite[Theorem 2]{DBLP:journals/jar/Platzer08}.
Furthermore, by our relative completeness proof \cite[Theorem 3]{DBLP:journals/jar/Platzer08}, hybrid systems verification, which is not a semidecidable problem \cite{DBLP:conf/lics/Henzinger96}, reduces completely to proving properties of differential equations.
Semantic validity defines the reference what properties are actually invariants, but for verification purposes we need a computable approximation that produces actual evidence in the form of a proof.

One simple but computable proof rule is \emph{differential weakening}:

\centerline{
\begin{calculus}
  \cinferenceRule[diffweak|$DW$]{differential weakening}
  {\linferenceRule[sequent]
    {\lsequent{\ivr}{\inv}}
    {\lsequent{\inv}{\dbox{\hevolvein{\D{x}=\theta}{\ivr}}{\inv}}}
  }{}
\end{calculus}
}%

\noindent
This rule is obviously sound, because the system \m{\hevolvein{\D{x}=\theta}{\ivr}}, by definition, can never leave $\ivr$, hence, if $\ivr$ implies $\inv$, then $\inv$ is an invariant, no matter what \m{\hevolve{\D{x}=\theta}} does.
Unfortunately, this simple proof rule cannot prove very interesting properties, because it only works when $\ivr$ is very informative.
It can, however, be useful in combination with stronger proof rules (e.g., the differential cuts that we discuss later).

\paragraph{Differential Invariants}
As a proof rule for fundamental invariance properties of differential equations, we have identified the following rule, called differential induction \cite{DBLP:journals/logcom/Platzer10,DBLP:conf/cav/PlatzerC08,DBLP:journals/fmsd/PlatzerC09}.
All premier proof principles for discrete loops are based on some form of induction.
Differential induction defines induction for differential equations.
It resembles induction for discrete loops but works for differential equations instead and uses a \emph{differential formula} (\m{\subst[\D{\inv}]{\D{x}}{\theta}}, which we develop below) for the induction step.

\centerline{
\begin{calculus}
  \cinferenceRule[diffind|$DI$]{differential invariant}
  {\linferenceRule[sequent]
    {\lsequent{\ivr}{\subst[\D{\inv}]{\D{x}}{\theta}}}
    {\lsequent{\inv}{\dbox{\hevolvein{\D{x}=\theta}{\ivr}}{\inv}}}
  }{}
\end{calculus}
}%
\noindent
This \emph{differential induction} rule is a natural induction principle for differential equations.
The difference compared to ordinary induction for discrete loops is that the evolution domain region~$\ivr$ is assumed in the premise (because the continuous evolution is not allowed to leave its evolution domain region) and that the induction step uses the differential formula \m{\subst[\D{\inv}]{\D{x}}{\theta}} corresponding to the differential equation \m{\hevolve{\D{x}=\theta}} and formula $\inv$ instead of a statement that the loop body preserves the invariant.
Intuitively, the \emph{differential formula} \m{\subst[\D{\inv}]{\D{x}}{\theta}} captures the infinitesimal change of formula $\inv$ over time along \m{\hevolve{\D{x}=\theta}}, and expresses the fact that $\inv$ is only getting more true when following the differential equation \m{\hevolve{\D{x}=\theta}}.
The semantics of differential equations is defined in a mathematically precise but computationally intractable way using analytic differentiation and limit processes at infinitely many points in time.
The key point about differential invariants is that they replace this precise but computationally intractable semantics with a computationally effective, algebraic and syntactic total derivative \m{\D{\inv}} along with simple substitution of differential equations.
Still. the valuation of the resulting computable formula \m{\subst[\D{\inv}]{\D{x}}{\theta}} along differential equations coincides with analytic differentiation.
It is defined as follows.

{\newcommand{\der}[1]{\DD{(#1)}}%
\begin{definition}%
  The operator~\m{\DD{}} that is defined as follows on terms is called \emph{syntactic (total) derivation}\index{derivation!syntactic}:
  \begin{subequations}
  \begin{align}
    \der{r} & = 0
      \hspace{2.1cm}\text{for numbers}~r\in\rationals
    \label{eq:Dconstant}\\
    \der{x} & =  \D{x}
      \hspace{2cm}\text{for variable}~x\label{eq:Dpolynomial}\\
    \der{a+b} & = \der{a} + \der{b}
    \label{eq:Dadditive}\\
    \der{a\cdot b} & = \der{a}\cdot b + a\cdot\der{b}
    \label{eq:DLeibniz}
  \end{align}
  \label{eq:Dterm}
  \end{subequations}
  \index{differential!symbol}%
  We extend it to (quantifier-free) first-order real-arithmetic formulas $F$ as follows:
  \begin{subequations}
  \begin{align}
    \der{F\land G} &\,\mequiv\, \der{F} \land \der{G}\\
    \der{F\lor G} &\,\mequiv\, \der{F} \land \der{G}
    \\
    \der{a\geq b} &\,\mequiv\, \der{a} \geq \der{b}
    \hspace{1cm}\text{accordingly for \m{<,>,\leq,=}}
  \end{align}
  \label{eq:Dformula}
  \end{subequations}
  We abbreviate \m{\der{\inv}} by \m{\D{\inv}} and define \m{\subst[\D{\inv}]{\D{x}}{\theta}} to be the result of substituting $\theta$ for $\D{x}$ in \m{\D{\inv}}.
\end{definition}
The conditions~\rref{eq:Dterm} define a derivation operator on terms that~\rref{eq:Dformula} lifts conjunctively to logical formulas.
It is important for the soundness of \irref{diffind} to define \m{\D{(F\lor G)}} as \m{\D{F} \land \D{G}}, because both subformulas need to satisfy the induction step, it is not enough if $F$ satisfies the induction step and $G$ holds initially; see \cite{DBLP:journals/logcom/Platzer10} for details. 
We assume that formulas use dualities like \m{\lnot(a\geq b) \mequiv a<b} to avoid negations; see \cite{DBLP:journals/logcom/Platzer10} for a discussion of this and the $\neq$ operator (which could be defined by \m{\der{a\neq b} \,\mequiv\, \der{a}=\der{b}} if needed) and division (which is easy to add).
For a discussion why this definition of differential invariants gives a sound approach and many other attempts would be unsound, we refer to previous work \cite{DBLP:journals/logcom/Platzer10,Platzer10}.
In the interest of a self-contained presentation, we give a soundness proof for differential invariants in \rref{app:sound}.
}

\begin{wrapfigure}[10]{r}{3.6cm}
  \vspace{-\baselineskip}  %
  \includegraphics[width=3.6cm]{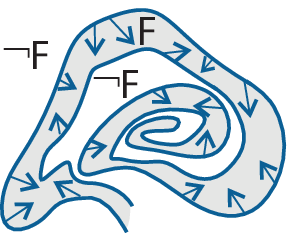} 
  \caption{Differential invariant~$F$}
  \label{fig:diffind}
\end{wrapfigure}
The basic idea behind rule \irref{diffind} is that the premise of \irref{diffind} shows that the total derivative~\m{\D{\inv}} holds within evolution domain~$\ivr$ when substituting the differential equations \m{\D{x}=\theta} into~$\D{\inv}$.
If~$\inv$ holds initially (antecedent of conclusion), then~$\inv$ itself stays true (succedent of conclusion).
Intuitively, the premise gives a condition showing that, within~$\ivr$, the total derivative~$\D{\inv}$ along the differential constraints is pointing inwards or transversally to~$\inv$ but never outwards to~$\lnot \inv$; see \rref{fig:diffind}.
Hence, if we start in~$\inv$ and, as indicated by~$\D{\inv}$, the local dynamics never points outside~$\inv$, then the system always stays in~$\inv$ when following the dynamics.
Observe that, unlike \m{\D{\inv}}, the premise of \irref{diffind} is a well-formed formula, because all differential expressions are replaced by non-differential terms when forming~\m{\subst[\D{\inv}]{\D{x}}{\theta}}.
Recall that we assume for simplicity that the (vectorial) differential equation \m{\hevolve{\D{x}=\theta}} mentions all variables of $\inv$.
It is possible to give a meaning also to the differential formula \m{\D{\inv}} itself in differential states \cite{DBLP:journals/logcom/Platzer10}, but this is not necessary for the questions we address in this paper.
Crucial for soundness, however, is the result that the valuation of syntactic derivatives along differential equations coincides with analytic differentiation; see \rref{app:sound} for a proof.

{\newcommand{\crf}{c}%
\newcommand{\der}[2][]{\subst[#2']{\D{x}}{\theta}}%
\newcommand{\If}{\DALint[flow=\varphi]}%
\begin{lemma}[Derivation lemma] \label{lem:derivationLemma}
  \newcommand{\Iff}{\iconcat[state=\varphi(t)]{\I}}%
  \newcommand{\Ifz}{\iconcat[state=\varphi(\zeta)]{\I}}%
  Let~$\hevolvein{\D{x}=\theta}{\ivr}$ be a differential equation with evolution domain constraint~$\ivr$ and
  let~$\iget[flow]{\If}:[0,r]\to(V\to\reals)$ be a corresponding solution of duration~$r>0$.
  Then for all terms~$\crf$ and all~$\zeta\in[0,r]$:
  \begin{displaymath}
    \D[t]{\,{\ivaluation{\Iff}{\crf}}}(\zeta) = \ivaluation{\Ifz}{\der{\crf}}
    \enspace.
  \end{displaymath}
  In particular,~\m{\ivaluation{\Iff}{\crf}} is continuously differentiable.
\end{lemma}
}

For the purposes of this paper, it is useful to note that~\m{\D{\inv}} can be written as the conjunction of total derivations of all atomic formulas in~$\inv$,
and \m{\subst[\D{\inv}]{\D{x}}{\theta}} is the result of substituting the (vectorial) differential equation \m{\hevolve{\D{x}=\theta}} into~\m{\D{\inv}}:
  \begin{align*}
    \D{\inv}
    &~\mequiv~
    \landfold_{(b\sim c) ~\text{in}~ \inv} \left(
      \left(\sum_{i=1}^{n} \Dp[x_i]{b} {\D{x_i}}\right)
      \sim
      \left(\sum_{i=1}^{n} \Dp[x_i]{c} {\D{x_i}}\right)
    \right)
    \\
    \subst[\D{\inv}]{\D{x}}{\theta}
    &~\mequiv~
    \landfold_{(b\sim c) ~\text{in}~ \inv} \left(
      \left(\sum_{i=1}^{n} \Dp[x_i]{b} {\theta_i}\right)
      \sim
      \left(\sum_{i=1}^{n} \Dp[x_i]{c} {\theta_i}\right)
    \right)
  \end{align*}
These conjunctions are over all atomic subformulas $b\sim c$ of $F$ for any \m{{\sim}\in\{=,\geq,>,\leq,<\}}.

\begin{definition}[Differential invariant]
  The (quantifier-free) formula $\inv$ of first-order real arithmetic is a \emph{differential invariant} of the system \m{\hevolvein{\D{x}=\theta}{\ivr}} if rule \irref{diffind} proves \m{\inv \limply \dbox{\hevolvein{\D{x}=\theta}{\ivr}}{\inv}}, i.e., if the formula \m{\lsequent{\ivr}{\subst[\D{\inv}]{\D{x}}{\theta}}} is provable.
\end{definition}
Since proof rule \irref{diffind} is sound, i.e., every provable formula is valid, every differential invariant is an invariant.
Since the first-order real arithmetic formula \m{\ivr\limply\subst[\inv]{\D{x}}{\theta}} is defined by a simple differential-algebraic computation, which can be performed symbolically, it is decidable whether $\inv$ is a differential invariant of a system \m{\hevolvein{\D{x}=\theta}{\ivr}} using quantifier elimination \cite{tarski_decisionalgebra51}.

The big advantage of rule \irref{diffind} is that it can be used to prove properties of differential equations without having to know their solution (solutions may fall outside decidable classes of arithmetic, may not be computable, or may not even exist in closed form).
A differential invariant $\inv$ is an implicit proof certificate for the validity of \m{\inv \limply \dbox{\hevolvein{\D{x}=\theta}{\ivr}}{\inv}}, because it establishes the same truth by a formal proof but does not need an explicit closed-form solution.
Furthermore, because differential equations are simpler than their solutions (which is part of the representational power of differential equations) and differential invariants are defined by differentiation (unlike solutions which are ultimately defined by integration), the differential induction rule \irref{diffind} is computationally attractive.
\begin{example} \label{ex:rotate}
The rotational dynamics \m{\hevolve{\D{x}=y\syssep\D{y}=-x}} is complicated in that the solution involves trigonometric functions, which are generally outside decidable classes of arithmetic.
Yet, we can easily prove interesting properties about it using \irref{diffind} and decidable polynomial arithmetic.
For instance, we can prove the simple property that \m{x^2+y^2\geq p^2} is a differential invariant of the dynamics using the following formal proof:
\renewcommand{\arraystretch}{1.3}%
\begin{sequentdeduction}[array]
 \linfer[diffind]
 {\linfer
   {\linfer[qear]
     {\lclose}
     {\lsequent{}{2xy+2y(-x)\geq0}}
   }
   {\lsequent{}{\subst[(2x\D{x}+2y\D{y}\geq0)]{\D{x}}{y}\subst[\,]{\D{y}}{-x}}}
 }
 {\lsequent{x^2+y^2\geq p^2}{\dbox{\hevolve{\D{x}=y\syssep\D{y}=-x}}{x^2+y^2\geq p^2}}}
\end{sequentdeduction}
Differential invariant proofs of more involved properties of rotational and curved flight dynamics can be found in previous work \cite{DBLP:journals/logcom/Platzer10,DBLP:conf/fm/PlatzerC09,Platzer10}.
\end{example}

\begin{example} \label{ex:damposc}
  Consider the dynamics \m{\hevolve{\D{x}=y\syssep\D{y} = -\omega^2 x - 2d\omega y}} of the damped oscillator with the undamped angular frequency $\omega$ and the damping ratio $d$.
  General symbolic solutions of symbolic initial-value problems for this differential equation can become surprisingly difficult.
  Mathematica, for instance, produces a 6 line equation of exponentials just for one solution.
  A differential invariant proof, instead, is very simple:
  \renewcommand{\arraystretch}{1.4}%
  \begin{sequentdeduction}[array]
  \linfer[diffind]
  {\linfer
    {\linfer[qear]
      {\lclose}
      {\lsequent{\omega\geq0 \land d\geq0} {2\omega^2xy-2\omega^2xy - 4d\omega y^2 \leq0}}
    }
    {\lsequent{\omega\geq0 \land d\geq0} {\subst[(2\omega^2x\D{x} +2y\D{y} \leq0)]{\D{x}}{y}\subst[\,]{\D{y}}{-\omega^2x-2d\omega y}}}
  }
  {\lsequent{\omega^2x^2+y^2\leq c^2} {\dbox{\hevolvein{\D{x}=y\syssep\D{y} = -\omega^2 x - 2d\omega y}{(\omega\geq0 \land d\geq0)}}{\,\omega^2x^2+y^2\leq c^2}}}
  \end{sequentdeduction}
  Observe that rule \irref{diffind} directly makes the evolution domain constraint \m{\omega\geq0 \land d\geq0} available as an assumption in the premise, because the continuous evolution is never allowed to leave it.
  This is useful if we have a strong evolution domain constraint or can make it strong during the proof, which we consider in \rref{sec:diffcut}.
\end{example}
These are simple examples illustrating the power of differential invariants.
Differential invariants make it possible to come up with very simple proofs even for tricky dynamics.
Logical proofs with differential invariants have been the key enabling technique for the successful verification of case studies in air traffic, railway, automotive, and electrical circuit domains.
Yet, if the original formula is not a differential invariant, one has to find the right differential invariant $\inv$ like in \rref{eq:odeb-indirect}, and, in particular, the search space for automatic procedures needs to include differential invariants of the right form.
In this paper, we consider theoretical questions of how to trade-off deductive power with the size of the search space.
We will answer the question which restrictions of differential invariants reduce the deductive power and which do not.

Because the premise of \irref{diffind} is in the (decidable) first-order theory of real arithmetic, it is decidable whether a given  formula $\inv$ is a differential invariant of a given system \m{\hevolvein{\D{x}=\theta}{\ivr}}.
For example, we can easily decide that \m{x^2+y^2\geq p^2} is a differential invariant of the dynamics in \rref{ex:rotate} and that \m{\omega^2x^2+y^2\leq c^2} is a differential invariant of the dynamics in \rref{ex:damposc}, just by deciding the resulting arithmetic \cite{tarski_decisionalgebra51} in the proofs of those examples.

Similarly, when the user specifies a formula $\inv$ with extra parameters $a_1,\dots,a_n$, it is obviously decidable whether there is a choice for those parameters that makes $\inv$ a differential invariant of a given system \m{\hevolvein{\D{x}=\theta}{\ivr}}.
All we need to do to see why this is decidable, is to write appropriate quantifiers in front of the formulas; see \cite{DBLP:conf/cav/PlatzerC08,DBLP:journals/fmsd/PlatzerC09,Platzer10} for formal details.
This is a simple approach, but deceptively simple.
If we choose the wrong template, it still will not work.
Furthermore, the approach has a high computational complexity so that the choice of appropriate templates is crucial.
For instance, a degree 2 template with a single polynomial for \rref{ex:damposc} will result in a formula with at least 23 quantifiers.
A single polynomial degree 4 template, which is the degree of the differential invariant in \rref{ex:damposc}, will result in at least 128 quantifiers.
Quantifier elimination has doubly exponential lower bounds \cite{DBLP:journals/jsc/DavenportH88} and practical quantifier elimination implementations are doubly exponential in the number of variables.
We, thus, need to understand the structure of the search space well to choose the right differential invariants or templates and avoid practically infinite computations that try to solve problems with the wrong templates.

\paragraph{Differential Cuts}
In the case of loops, invariants can be assumed to hold before the loop body in the induction step.
It thus looks tempting to suspect that rule \irref{diffind} could be improved by assuming the differential invariant $\inv$ in the antecedent of the premise:

\centerline{
\begin{calculus}
  \cinferenceRule[diffindunsound|$DI_{??}$]{unsound}
  {\linferenceRule[sequent]
    {\lsequent{\ivr\land\inv}{\subst[\D{\inv}]{\D{x}}{\theta}}}
    {\lsequent{\inv}{\dbox{\hevolvein{\D{x}=\theta}{\ivr}}{\inv}}}
  }{}
\end{calculus}
\quad\text{sound?}
}%
\noindent
After all, we really only care about staying safe when we are still safe.
But implicit properties of differential equations are a subtle business.
Assuming $\inv$ like in rule \irref{diffindunsound} would, in fact, be unsound, as the following simple counterexample shows, which ``proves'' an invalid property using \irref{diffindunsound}:
{\renewcommand{\arraystretch}{1.3}%
\begin{equation}
\begin{minipage}{0.8\columnwidth}
\begin{sequentdeduction}[array]
 \linfer[diffindunsound]
 {\linfer
   {\linfer
     {\lclose[\ast~\text{(unsound)}]}
     {\lsequent{}{-(x-y)^2\geq0\limply-2(x-y)(1-y)\geq0}}
   }
   {\lsequent{}{-(x-y)^2\geq0\limply\subst[(-2(x-y)(\D{x}-\D{y})\geq0)]{\D{x}}{1}\subst[\,]{\D{y}}{y}}}
 }
 {\lsequent{-(x-y)^2\geq0}{\dbox{\hevolve{\D{x}=1\syssep\D{y}=y}}{(-(x-y)^2\geq0)}}}
\end{sequentdeduction}
\end{minipage}
\label{eq:diffind-restrict}
\end{equation}
Especially, it would also be unsound to restrict the premise of \irref{diffind} to the boundary $\partial\inv$ of $\inv$.
The problem causing this unsoundness is circular reasoning and the fact that derivatives are only defined in domains with non-empty interior.
In the beginning, we only know invariant $F$ to hold at a single point (antecedent of conclusion of \irref{diffind}). Now if we assume $F$ to hold in some domain for the induction step (premise of \irref{diffind}), then, initially, we actually only know that $F$ holds in a region with an empty interior.
This is not sufficient to conclude anything based on derivatives, because these are not defined unless there is a neighborhood around the point (with non-empty interior).
Thus, the reasoning in \rref{eq:diffind-restrict} assumes more than it has proven already, which explains the unsoundness, and the unsoundness of the rule \irref{diffindunsound}.
For an analysis under which circumstances extra assumption $\inv$ could be assumed in the premise without losing soundness, we refer to previous work \cite{DBLP:journals/logcom/Platzer10}.
}%

Instead, we have come up with a complementary proof rule for \emph{differential cuts} \cite{DBLP:journals/logcom/Platzer10,DBLP:conf/cav/PlatzerC08,DBLP:journals/fmsd/PlatzerC09} that can be used to strengthen assumptions in a sound way:

\centerline{
\begin{calculus}
  \cinferenceRule[diffcut|$DC$]{differential cut}
  {\linferenceRule[sequent]
    {\lsequent{\inv}{\dbox{\hevolvein{\D{x}=\theta}{\ivr}}{C}}
    &&\lsequent{\inv}{\dbox{\hevolvein{\D{x}=\theta}{(\ivr\land C)}}{\inv}}}
    {\lsequent{\inv}{\dbox{\hevolvein{\D{x}=\theta}{\ivr}}{\inv}}}
  }{}
\end{calculus}
}%
\noindent
The differential cut rule works like a cut, but for differential equations.
In the right premise, rule \irref{diffcut} restricts the system evolution to the subdomain \m{\ivr \land C} of $\ivr$, which changes the system dynamics but is a pseudo-restriction, because the left premise proves that $C$ is an invariant anyhow (e.g. using rule \irref{diffind}).
Note that rule \irref{diffcut} is special in that it changes the dynamics of the system (it adds a constraint to the system evolution domain region), but it is still sound, because this change does not reduce the reachable set.
The benefit of rule \irref{diffcut} is that $C$ will (soundly) be available as an extra assumption for all subsequent \irref{diffind} uses on the right premise (see, e.g., the use of the evolution domain constraint in \rref{ex:damposc}).
In particular, the differential cut rule \irref{diffcut} can be used to strengthen the right premise with more and more auxiliary differential invariants $C$ that will be available as extra assumptions on the right premise, once they have been proven to be differential invariants in the left premise.

Using this differential cut process repeatedly has turned out to be extremely useful in practice and even simplifies the invariant search, because it leads to several simpler properties to find and prove instead of a single complex property \cite{DBLP:conf/cav/PlatzerC08,DBLP:journals/fmsd/PlatzerC09,Platzer10}.
But is it necessary in theory or just convenient in practice?
Should we be searching for proofs without differential cuts or should we always conduct proof search including differential cuts?
One central question that we answer in this paper is whether there is a differential cut elimination theorem showing that \irref{diffcut} is admissible, or whether differential cuts are fundamental, because the addition of rule \irref{diffcut} extends the deductive power.

\paragraph{Prelude}
As a prelude to all subsequent (meta-)proofs, we ignore constant polynomials in differential invariants, because they do not contribute to the proof.
For example, \m{5\geq0} and \m{0=0} are trivially true (do not contribute) and \m{0\geq1} and \m{2=0} are trivially false (not implied by any satisfiable precondition).
We, thus, do not need to consider them for provability purposes, because they do not constitute useful differential invariants.
That is, whenever there is a proof using those trivial differential invariants, there also is a shorter proof not using them.
Likewise, we do not need to consider polynomial conditions like \m{x^2+1\geq0} that are trivially true or trivially false \m{-x^2>0}.

Furthermore, the subsequent proofs will go at an increasing pace.
The first proofs will show elementary steps in detail, while subsequent proofs will proceed with a quicker pace and use the same elementary decompositions as previous proofs.
One of the tricky parts in the proofs is coming up with the right counterexample to an inclusion or proving that there is none.
The other tricky part is to show deductive power separation properties, i.e., that a valid formula cannot be proven using a given subset of the proof rules, which is a proof about infinitely many formal proofs.

This is similar to the fact that, in algebra, it is easier to prove that two structures are isomorphic than to prove that they are not.
That they are isomorphic can be proven by constructing an isomorphism and proving that it satisfies all required properties.
But proving that they are non-isomorphic requires a proof that \emph{every} function between the two structures violates at least one of the properties of an isomorphism.
Those proofs work by identifying a characteristic that is preserved by isomorphisms (e.g., dimension of vector spaces) but that the two structures under consideration do not agree on.
We identify corresponding characteristics for the separation properties of deductive power.

\section{Equivalences of Differential Invariants}

First, we study whether there are equivalence transformations that preserve differential invariance.
Every equivalence transformation that we have for differential invariant properties helps us with structuring the proof search space and also helps simplifying meta-proofs.
\begin{lemma}[Differential invariants and propositional logic] \label{lem:diffind-prop-equiv}
  Differential invariants are invariant under propositional equivalences.
  That is, if \m{F \lbisubjunct G} is an instance of a propositional tautology then $F$ is a differential invariant of \m{\hevolvein{\D{x}=\theta}{\ivr}} if and only if $G$ is.
\end{lemma}
\begin{proof}
  Let $F$ be a differential invariant of a differential equation system \m{\hevolvein{\D{x}=\theta}{\ivr}} and let $G$ be a formula such that \m{F \lbisubjunct G} is an instance of a propositional tautology.
  Then $G$ is a differential invariant of \m{\hevolvein{\D{x}=\theta}{\ivr}}, because of the following formal proof:
\renewcommand{\arraystretch}{1.4}%
\begin{sequentdeduction}[array]
  \linfer%
    {\linfer[diffind]
      {\linfer%
        {\lclose}
        {\lsequent{\ivr}{\subst[\D{G}]{\D{x}}{\theta}}}
      }
      {\lsequent{G}{\dbox{\hevolvein{\D{x}=\theta}{\ivr}}{G}}}
    }
    {\lsequent{F}{\dbox{\hevolvein{\D{x}=\theta}{\ivr}}{F}}}
\end{sequentdeduction}
The bottom proof step is easy to see using \rref{eq:odeb-indirect}, because precondition $F$ implies the new precondition $G$ and postcondition $F$ is implied by the new postcondition $G$ propositionally.
Subgoal \m{\lsequent{\ivr}{\subst[\D{G}]{\D{x}}{\theta}}} is provable, because \m{\lsequent{\ivr}{\subst[\D{F}]{\D{x}}{\theta}}} is provable and $\D{G}$ is defined as a conjunction over all literals of $G$.
The set of literals of $G$ is identical to the set of literals of $F$, because the literals do not change by using propositional tautologies.
Furthermore, we assumed a propositionally complete base calculus (e.g., \rref{app:proof-rules}).
\qedhere
\end{proof}
In subsequent proofs, we can use propositional equivalence transformations by \rref{lem:diffind-prop-equiv}.
In the following, we will also implicitly use equivalence reasoning for pre- and postconditions as we have done in \rref{lem:diffind-prop-equiv}.
Because of \rref{lem:diffind-prop-equiv}, we can, without loss of generality, work with arbitrary propositional normal forms for proof search.

Not all logical equivalence transformations carry over to differential invariants.
Differential invariance is not necessarily preserved under real arithmetic equivalence transformations.

\begin{lemma}[Differential invariants and arithmetic] \label{lem:diffind-nequiv}
  Differential invariants are \emph{not} invariant under equivalences of real arithmetic.
  That is, if \m{F \lbisubjunct G} is an instance of a first-order real arithmetic tautology then $F$ may be a differential invariant of \m{\hevolvein{\D{x}=\theta}{\ivr}} yet $G$ may not.
\end{lemma}
\proof
\renewcommand{\arraystretch}{1.2}%
\def\tmpa{5}%
There are two formulas that are equivalent over first-order real arithmetic but, for the same differential equation, one of them is a differential invariant, the other one is not (because their differential structures differ).
Since $\tmpa\geq0$, the formula \m{x^2\leq\tmpa^2} is equivalent to \m{-\tmpa\leq x\land x\leq \tmpa} in first-order real arithmetic.
Nevertheless, \m{x^2\leq\tmpa^2} is a differential invariant of \m{\hevolve{\D{x}=-x}} by the following formal proof:
\begin{sequentdeduction}[array]
  \linfer[diffind]
    {\linfer%
      {\linfer[qear]
        {\lclose}
        {\lsequent{}{-2x^2\leq0}}
      }
      {\lsequent{}{\subst[(2x\D{x}\leq0)]{\D{x}}{-x}}}
    }
    {\lsequent{x^2\leq\tmpa^2}{\dbox{\hevolve{\D{x}=-x}}{x^2\leq\tmpa^2}}}
\end{sequentdeduction}
but \m{-\tmpa\leq x \land x\leq\tmpa} is not a differential invariant of \m{\hevolve{\D{x}=-x}}:
\begin{sequentdeduction}[array]
  \linfer[diffind]
    {\linfer%
      {\linfer%
        {\lclose[\text{not valid}]}
        {\lsequent{}{0\leq-x\land-x\leq0}}
      }
      {\lsequent{}{\subst[(0\leq\D{x}\land\D{x}\leq0)]{\D{x}}{-x}}}
    }
    {\lsequent{-\tmpa\leq x\land
        x\leq\tmpa}{\dbox{\hevolve{\D{x}=-x}}{(-\tmpa\leq x \land
          x\leq\tmpa)}}}\rlap{\hbox to 103 pt{\hfill\qEd}}
\end{sequentdeduction}\smallskip

\noindent When we want to prove the property in the proof of \rref{lem:diffind-nequiv}, we need to use the principle \rref{eq:odeb-indirect} with the differential invariant \m{\inv\mequiv x^2\leq5^2} and cannot use \m{-5\leq x \land x\leq5}.

By \rref{lem:diffind-nequiv}, we cannot just use arbitrary equivalences when investigating differential invariance, but have to be more careful.
Not just the \emph{elementary real arithmetical equivalence} of having the same set of satisfying assignments matters, but also the differential structures need to be compatible.
Some equivalence transformations that preserve the solutions still destroy the differential structure.
It is the equivalence of \emph{real differential structures} that matters.
Recall that differential structures are defined locally in terms of the behavior in neighborhoods of a point, not the point itself.

\rref{lem:diffind-nequiv} illustrates a notable point about differential equations.
Many different formulas characterize the same set of satisfying assignments.
But not all of them have the same differential structure.
Quadratic polynomials have inherently different differential structure than linear polynomials even when they have the same set of solutions over the reals.
The differential structure is a more fine-grained information.
This is similar to the fact that two elementary equivalent models of first-order logic can still be non-isomorphic.
Both the set of satisfying assignments and the differential structure matter for differential invariance.
In particular, there are many formulas with the same solutions but different differential structures.
The formulas \m{x^2\geq0} and 
\m{x^6+x^4-16x^3+97x^2-252x+262\geq0} %
have the same solutions (all of $\reals$), but very different differential structure; see \rref{fig:diffstructure}.
The first two rows in \rref{fig:diffstructure} correspond to the polynomials from the latter two cases.
The third row is a structurally different degree 6 polynomial with again the same set of solutions ($\reals$) but a rather different differential structure.
The differential structure also depends on what value $\D{x}$ assumes according to the differential equation.
\begin{figure}[tb]
  \centering
  \includegraphics[height=4cm]{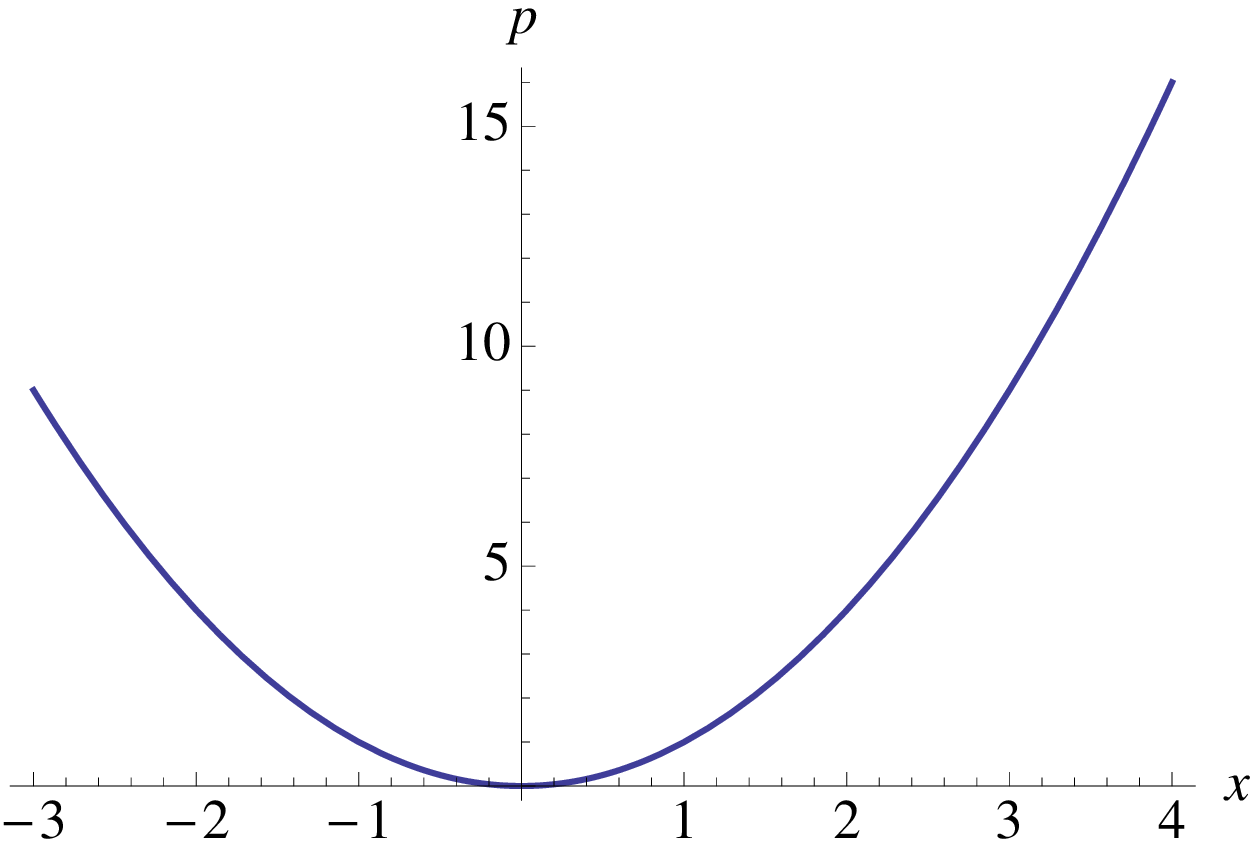}
  \qquad
  \includegraphics[height=4cm]{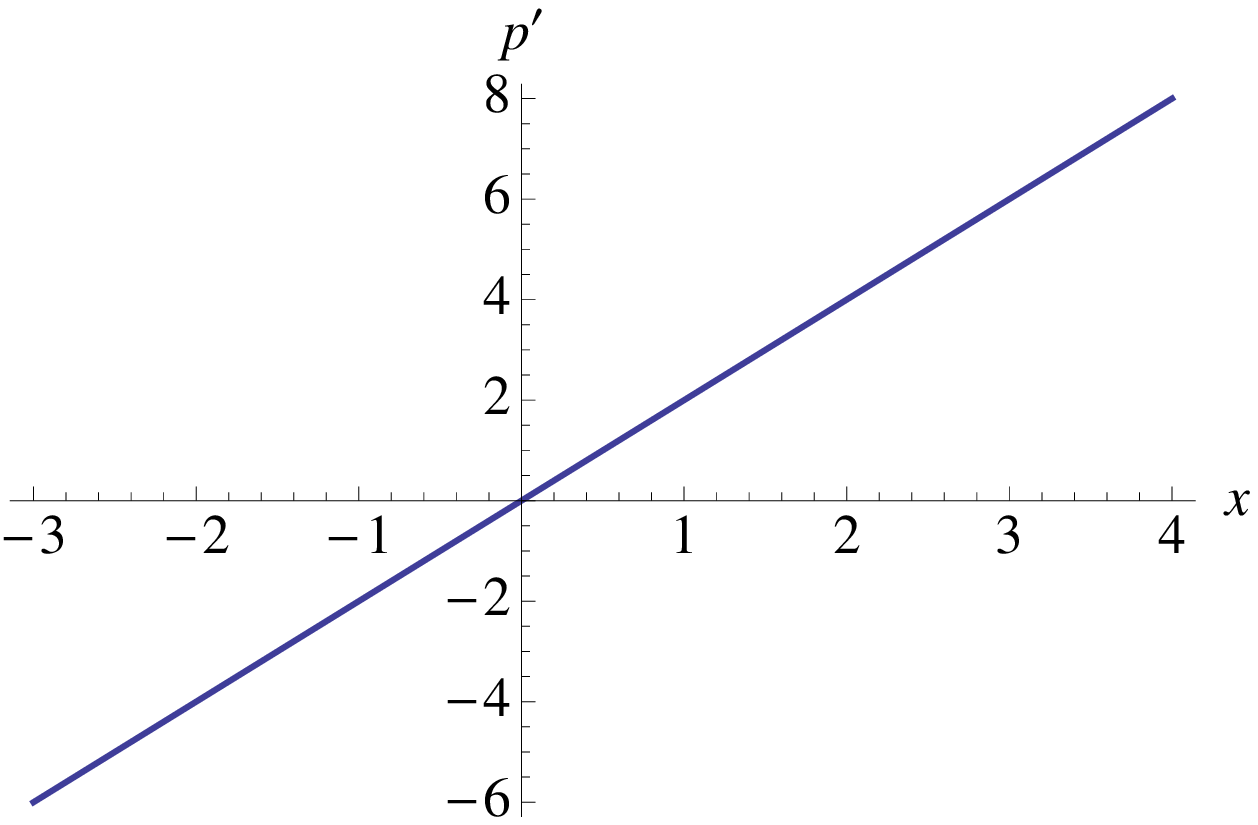}

  \bigskip

  \includegraphics[height=4cm]{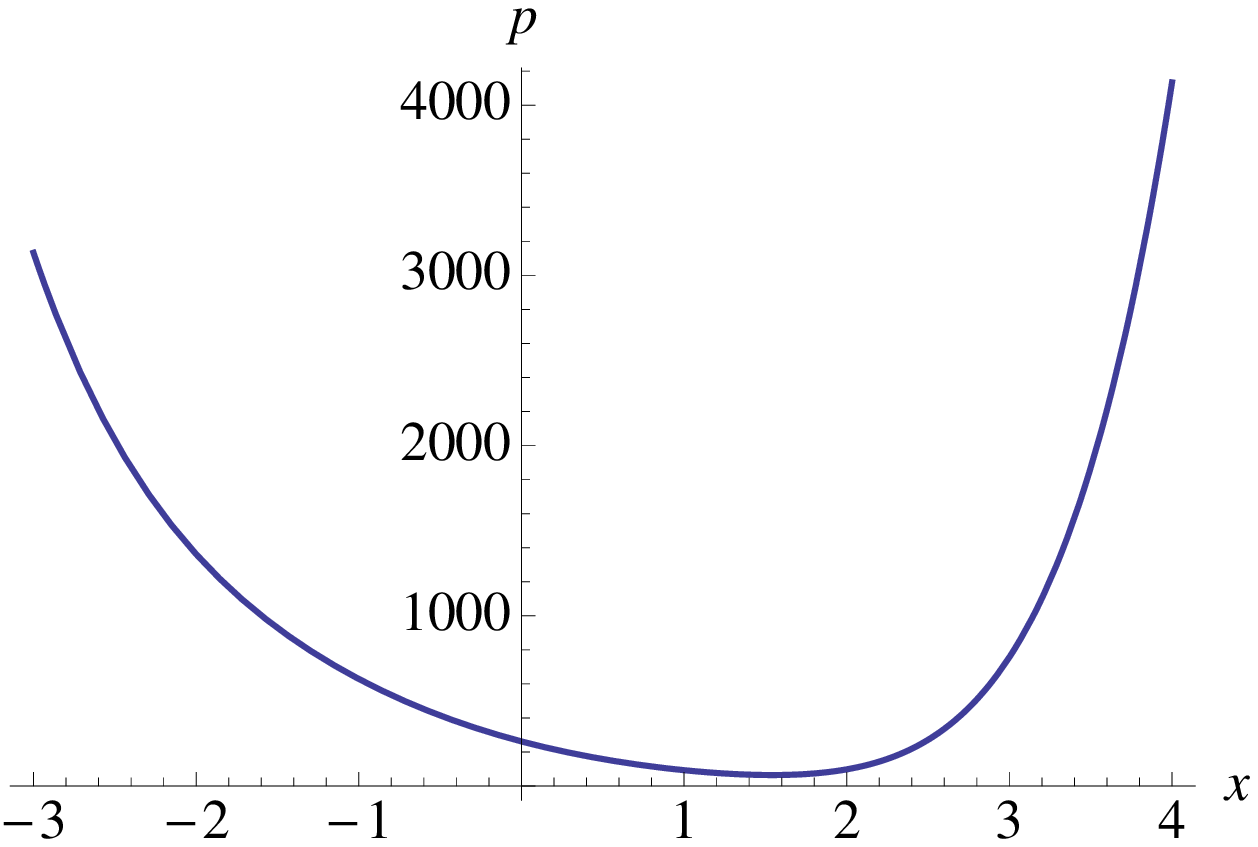}
  \qquad
  \includegraphics[height=4cm]{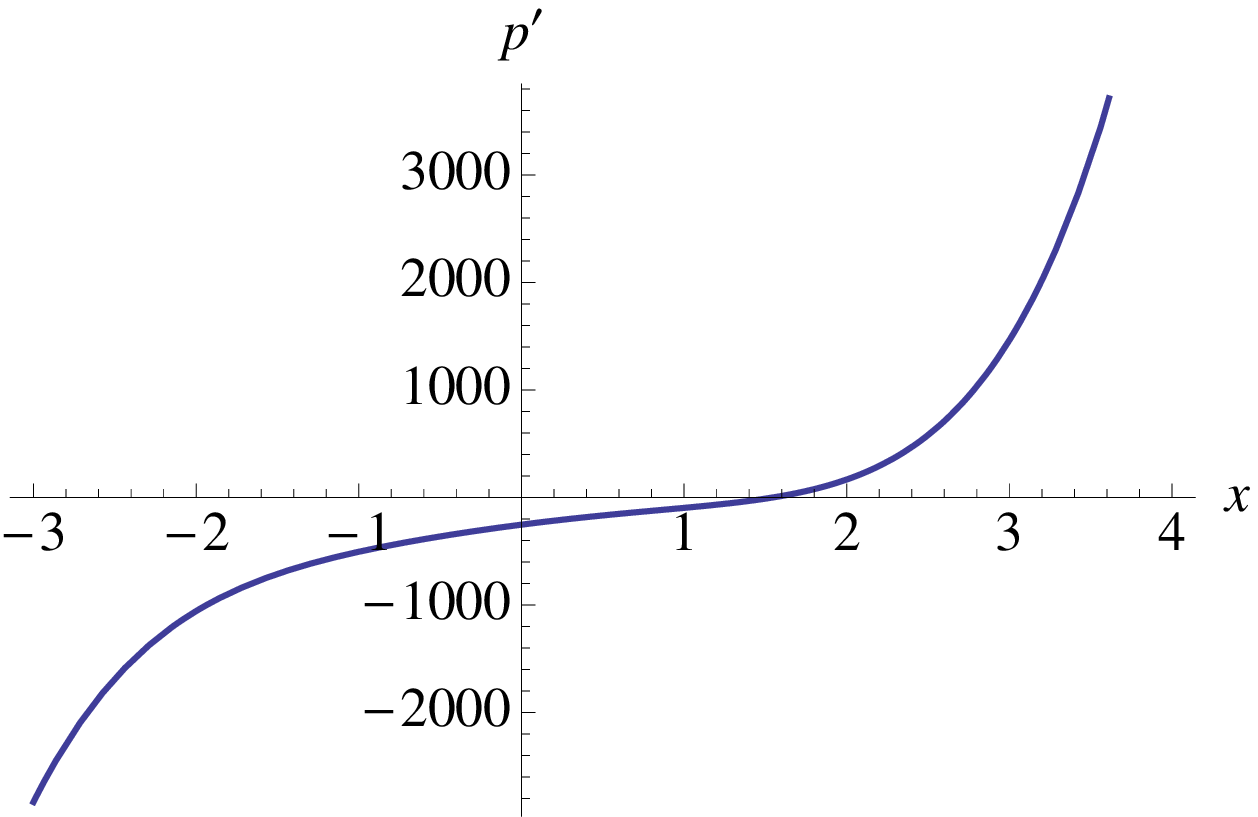}

  \bigskip

  \includegraphics[height=4cm]{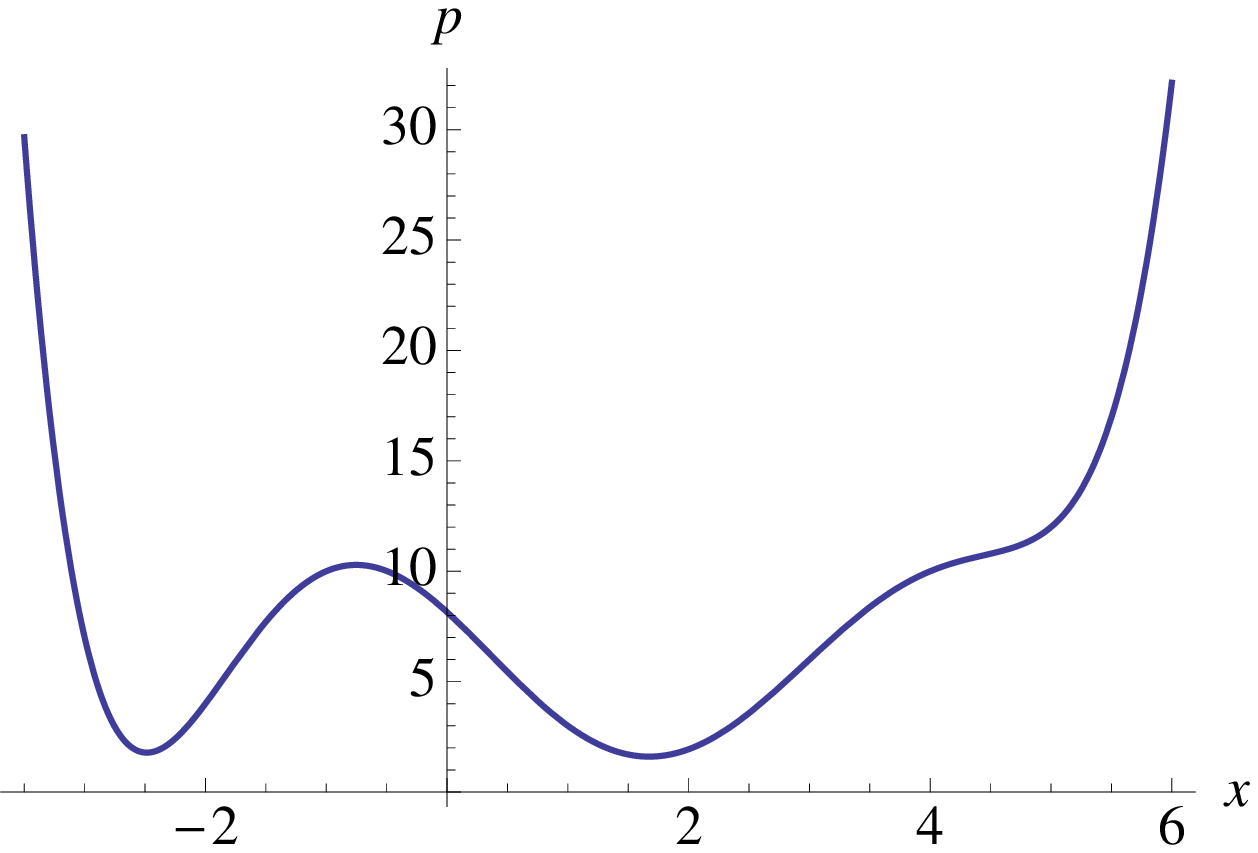}
  \qquad
  \includegraphics[height=4cm]{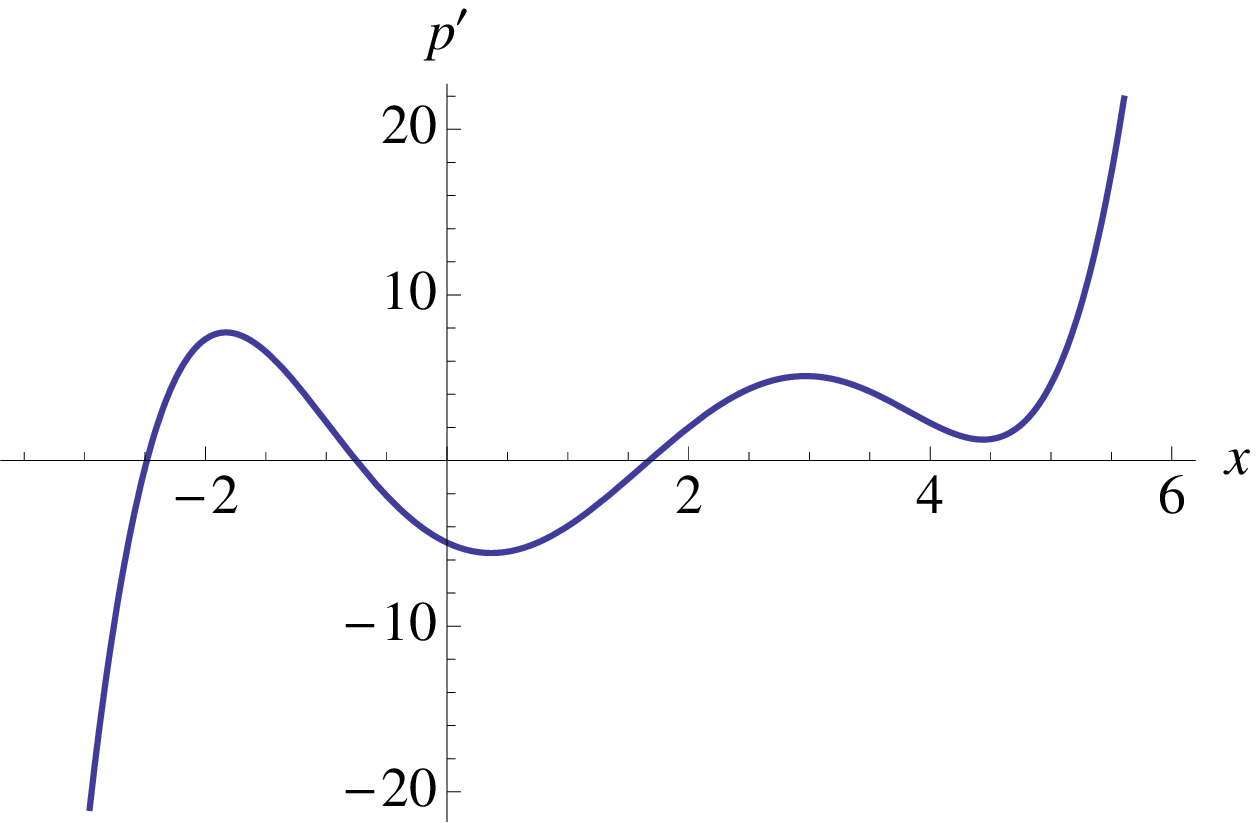}
  \caption{Equivalent solutions ($p\geq0$ on the left) with different differential structure ($\D{p}$ plotted on the right)}
  \label{fig:diffstructure}
\end{figure}
\rref{fig:diffstructure} illustrates that $\D{p}$ alone can already have a very different characteristic even if the respective sets of satisfying assignments of \m{p\geq0} are identical.

We can, however, always normalize all atomic subformulas to have right-hand side 0, that is, of the form \m{p=0, p\geq0}, or \m{p>0}.
For instance, \m{p\leq q} is a differential invariant if and only if \m{q-p\geq0} is, because \m{p\leq q} is equivalent (in first-order real arithmetic) to \m{q-p\geq0} and, moreover, for any variable $x$ and term $\theta$, \m{\subst[(\D{p}\leq\D{q})]{\D{x}}{\theta}} is equivalent to \m{\subst[(\D{q}-\D{p}\geq0)]{\D{x}}{\theta}} in first-order real arithmetic.

\section{Relations of Differential Invariant Classes} \label{sec:chart}

We study the relations of classes of differential invariants in terms of their relative deductive power.
As a basis, we consider a propositional sequent calculus with logical cuts (which simplify glueing derivations together) and real-closed field arithmetic (we denote all uses by proof rule \irref{qear}); see \rref{app:proof-rules}.
By $\DI$ we denote the proof calculus that, in addition, has general differential invariants (rule \irref{diffind} with arbitrary quantifier-free first-order formula $\inv$) but no differential cuts (rule \irref{diffcut}).
For a set \m{\Omega \subseteq \{\geq,>,=,\land,\lor\}} of operators, we denote by \m{\DI[\Omega]} the proof calculus where the differential invariant $\inv$ in rule \irref{diffind} is further restricted to the set of formulas that uses only the operators in $\Omega$.
For example, \m{\DIp[=]} is the proof calculus that allows only and/or-combinations of equations to be used as differential invariants.
Likewise, \m{\DI[\geq]} is the proof calculus that only allows atomic weak inequalities \m{p\geq q} to be used as differential invariants.

We consider several classes of differential invariants and study their relations.
If $\mathcal{A}$ and $\mathcal{B}$ are two classes of differential invariants, we write \m{\mathcal{A} \leq \mathcal{B}} if all properties provable using differential invariants from $\mathcal{A}$ are also provable using differential invariants from $\mathcal{B}$.
We write \m{\mathcal{A} \not\leq \mathcal{B}} otherwise, i.e., when there is a valid property that can only be proven using differential invariants of \m{\mathcal{A}\setminus\mathcal{B}}.
We write \m{\mathcal{A} \mequiv \mathcal{B}} if \m{\mathcal{A} \leq \mathcal{B}} and \m{\mathcal{B} \leq \mathcal{A}}.
We write \m{\mathcal{A} < \mathcal{B}} if \m{\mathcal{A} \leq \mathcal{B}} and \m{\mathcal{B} \not\leq \mathcal{A}}.
Classes $\mathcal{A}$ and $\mathcal{B}$ are incomparable if \m{\mathcal{A} \not\leq \mathcal{B}} and \m{\mathcal{B} \not\leq \mathcal{A}}.

In this section, we analyze the proof calculus \DI with differential invariants but no differential cuts (\irref{diffcut}) or auxiliary variables, which we discuss in Sections~\ref{sec:diffcut} and~\ref{sec:diffaux}, respectively.
Our findings about classes of differential invariants are summarized in \rref{fig:diffind-chart} on p.\,\pageref{fig:diffind-chart}.
We prove these relations in the remainder of this section.
These results hold in two ways.
The relations summarized in \rref{fig:diffind-chart} hold for single uses of differential invariants in any proof calculus and they also hold for arbitrary proofs with arbitrarily many uses of differential invariants in the proof calculus with the rules \irref{diffind}, \irref{diffweak}, \irref{genb} and the propositional and real arithmetic rules that we recall in \rref{app:proof-rules}.

First we recall a simple result from previous work showing that propositional operators do not change the deductive power of differential invariants in the purely equational case.
We have proven the following result in previous work; see \cite[Proposition~1]{DBLP:journals/logcom/Platzer10}.
We repeat a variation of the proof here, because it is instructive to understand what needs to be proved about the algebraic and differential structure of differential invariants.

\begin{proposition}[Equational deductive power \cite{DBLP:journals/logcom/Platzer10}] \label{prop:EDP}
  The deductive power of differential induction with atomic equations is identical to the deductive power of differential induction with propositional combinations of polynomial equations:
  That is, each formula is provable with propositional combinations of equations as differential invariants iff it is provable with only atomic equations as differential invariants:
  \[ \DI[=] \mequiv \DIp[=] \]
\end{proposition}
\proof
  Let \m{\D{x}=\theta} be the (vectorial) differential equation to consider.
  We show that every differential invariant that is a propositional combination~$\inv$ of polynomial equations is expressible as a single atomic polynomial equation (the converse inclusion is obvious).
  We can assume~$\inv$ to be in negation normal form by \rref{lem:diffind-prop-equiv} (recall that negations are resolved and $\neq$ does not appear).
  Then we reduce~$\inv$ inductively to a single equation using the following transformations:
  \begin{iteMize}{$\bullet$}
   \item If~$\inv$ is of the form~\m{p_1=p_2\lor q_1=q_2},
    then~$\inv$ is equivalent to the single equation\\
    \m{(p_1-p_2)(q_1-q_2)=0}.
    Furthermore, \m{\subst[\D{\inv}]{\D{x}}{\theta} \mequiv \subst[(\D{p_1}=\D{p_2} \land \D{q_1}=\D{q_2})]{\D{x}}{\theta}} directly implies
    \[\subst[\big(\D{((p_1-p_2)(q_1-q_2))}=0\big)]{\D{x}}{\theta} \mequiv \subst[\big((\D{p_1}-\D{p_2})(q_1-q_2) + (p_1-p_2)(\D{q_1}-\D{q_2})=0\big)]{\D{x}}{\theta}\]
   \item If~$\inv$ is of the form~\m{p_1=p_2\land q_1=q_2},
    then~$\inv$ is equivalent to the single equation \\
    \m{(p_1-p_2)^2+(q_1-q_2)^2=0}.
    Furthermore, \m{\subst[\D{\inv}]{\D{x}}{\theta} \mequiv \subst[\big(\D{p_1}=\D{p_2} \land \D{q_1}=\D{q_2}\big)]{\D{x}}{\theta}} implies
    \[\subst[\big(\D{\left((p_1-p_2)^2+(q_1-q_2)^2\right)} {=} 0\big)]{\D{x}}{\theta} \mequiv
    \subst[\big(2(p_1-p_2)(\D{p_1}-\D{p_2}) + 2(q_1-q_2)(\D{q_1}-\D{q_2})=0\big)]{\D{x}}{\theta}\eqno{\qEd}
    \]\medskip
  \end{iteMize}

\noindent Note that the polynomial degree increases quadratically by the reduction in \rref{prop:EDP}, but, as a trade-off, the propositional structure simplifies.
Consequently, differential invariant search for the equational case can either exploit propositional structure with lower degree polynomials or suppress the propositional structure at the expense of higher degrees.
Focusing exclusively on differential invariants with equations, however, reduces the deductive power.
For instance, the approach by Sankaranarayanan et al. \cite{DBLP:journals/fmsd/SankaranarayananSM08} uses only equations and does not support inequalities.

\begin{proposition}[Equational incompleteness] \label{prop:eq}
  The deductive power of differential induction with equational formulas is strictly less than the deductive power of general differential induction, because some inequalities cannot be proven with equations.
  \begin{align*}
   \DI[=] \mequiv \DIp[=] &< \DI \\
   \DI[\geq]  &\not\leq \DI[=] \mequiv \DIp[=] \\
   \DI[>]  &\not\leq \DI[=] \mequiv \DIp[=]
  \end{align*}
\end{proposition}

\proof
\def\tmpa{a}%
Consider any term $\tmpa>0$ (e.g., $5$ or $x^2+1$ or $x^2+x^4+2$).
The following formula is provable by differential induction with the weak inequality $x\geq0$:
\begin{sequentdeduction}[array]
  \linfer[diffind]
    {\linfer[qear]
      {\lclose}
      {\lsequent{}{\tmpa\geq0}}
    }
    {\lsequent{x\geq0}{\dbox{\hevolve{\D{x}=\tmpa}}{x\geq0}}}
\end{sequentdeduction}
It is not provable with an equational differential invariant.
An invariant of the form \m{p=0} has (Lebesgue) measure zero (except when $p$ is the 0 polynomial, where \m{p=0} is trivially equivalent to $\ltrue$ and then useless for a proof, because it provides no interesting information) and, thus, cannot describe the region \m{x\geq0} of non-zero (Lebesgue) measure, in which the system starts (precondition) and stays (postcondition).
More formally, any univariate polynomial $p$ that is zero on \m{x\geq0} is the zero polynomial and, thus, \m{p=0} cannot be equivalent to the half space \m{x\geq0}.
By the equational deductive power theorem~\ref{prop:EDP}, the above formula then is not provable with any Boolean combination of equations as differential invariant either.

We extend this argument to arbitrary proofs with an arbitrary number of uses of \irref{diffind}.
By \rref{lem:DIbase}, every \DI[=] proof of the above property uses a non-zero finite number of (nontrivial) equational differential invariants $p_i=0$, such that 
\m{\entails x\geq0 \limply \lorfold_i p_i=0}.
This is a contradiction, because finite unions of sets of measure zero have measure zero, yet \m{x\geq0} has non-zero measure.
Recall that the zero polynomial has non-zero measure but does not help (the second condition \m{0=0\limply x\geq0} of \rref{lem:DIbase} is not valid).

Similarly, the following formula is provable by differential induction with a strict inequality \m{x>0}, but, for the same reason of different measures (respectively infinitely many zeros), not provable by an invariant of the form \m{p=0} (or, using \rref{lem:DIbase}, not provable using any finite number of equational differential invariants):
\begin{sequentdeduction}[array]
  \linfer[diffind]
    {\linfer[qear]
      {\lclose}
      {\lsequent{}{\tmpa>0}}
    }
    {\lsequent{x>0}{\dbox{\hevolve{\D{x}=\tmpa}}{x>0}}}\rlap{\hbox to
      159 pt{\hfill\qEd}}
\end{sequentdeduction}\medskip

\noindent It might be tempting to think that at least equational postconditions (like those considered in \cite{DBLP:journals/fmsd/SankaranarayananSM08}) only need equational differential invariants for proving them.
But that is not the case either.
We show that there are even purely equational invariants that are only provable using inequalities, but not when using only equations as differential invariants.
\begin{theorem}[No equational closure] \label{thm:diffind-eq-leq}
  Some equational invariants of differential equations are only provable using an inequality as a differential invariant, but not using equational propositional logic for differential invariants.
  This equational invariant is not even provable using equational propositional logic and differential cuts.
\end{theorem}
\begin{proof}
The formula \m{x=0 \limply \dbox{\hevolve{\D{x}=-x}}{x=0}} is provable using \m{-x^2\geq0} as a differential invariant by the following simple formal proof:
\begin{sequentdeduction}[array]
  \linfer[diffind]
    {\linfer%
      {\linfer[qear]
        {\lclose}
        {\lsequent{}{2x^2\geq0}}
      }
      {\lsequent{}{\subst[(-2x\D{x}\geq0)]{\D{x}}{-x}}}
    }
    {\lsequent{-x^2\geq0}{\dbox{\hevolve{\D{x}=-x}}{(-x^2\geq0)}}}
\end{sequentdeduction}
We need to show that this formula cannot be proven using equations as differential invariants.
Suppose there was a differential invariant of the form \m{p=0} for a univariate polynomial $p$ of the form \m{\sum_{i=0}^n a_i x^i} in the only occurring variable $x$.
Then
\begin{enumerate}[(1)]
\item \label{case:diffind-eq-leq-base} \m{\entails p=0 \lbisubjunct x=0}, and
\item \label{case:diffind-eq-leq-diffind} \m{\entails \subst[\D{p}]{\D{x}}{-x} =0}, where
 \[\subst[\D{p}]{\D{x}}{-x} = \subst[\left(\sum_{i=1}^n i a_i x^{i-1} \D{x}\right)]{\D{x}}{-x} = - \sum_{i=1}^n i a_i x^i\]
\end{enumerate}
From condition~\ref{case:diffind-eq-leq-diffind}, we obtain that \m{a_1=a_2=\dots=a_0=0} by comparing coefficients.
Consequently, $p$ must be the constant polynomial $a_0$, not involving $x$.
Thus, the formula \m{p=0} is either trivially equivalent to $\ltrue$ (then it does not contribute to the proof) or equivalent to $\lfalse$ (then it is no consequence of the precondition).
Thus the only equational invariants of \m{x=0 \limply \dbox{\hevolve{\D{x}=-x}}{x=0}} are trivial (equivalent to $\ltrue$ or to $\lfalse$).
Consequently, that formula cannot be provable by an equational invariant, nor by a propositional combination of equations (because of \rref{prop:EDP}).

By \rref{lem:DIbase}, this result extends to arbitrary proofs in \DI[=] with an arbitrary number of uses of \irref{diffind}, because each equational differential invariant for the above dynamics is trivial.

This result still holds in the presence of differential cuts.
Differential cuts could use any formula $C$, but they only help closing the proof if the left premise of \irref{diffcut} can be proved using \irref{diffind}. Note that \irref{diffcut} could be used repeatedly, but does not on its own prove properties without \irref{diffind}, because the left premise of \irref{diffcut} needs to prove invariance of $C$ for the same dynamics.
We have shown above that the only equational differential invariants that \irref{diffcut} could successfully strengthen the formula with are trivial. They do not contain $x$, are equivalent to $\ltrue$ (and then do not contribute to the proof), or equivalent to $\lfalse$ (and then are not implied by the precondition).
\qedhere
\end{proof}

We show that, conversely, focusing on strict inequalities also reduces the deductive power, because equations are obviously missing and there is at least one proof where this matters.
That is, strict barrier certificates do not prove (nontrivial) closed invariants.
\begin{proposition}[Strict barrier incompleteness] \label{prop:gt}
  The deductive power of differential induction with strict barrier certificates (formulas of the form \m{p>0}) is strictly less than the deductive power of general differential induction.
  \begin{align*}
    \DI[>] &< \DI \\
    \DI[=] &\not\leq \DI[>]
  \end{align*}
\end{proposition}
\begin{proof}
The following formula is provable by equational differential induction:
\begin{sequentdeduction}[array]
  \linfer[diffind]
    {\linfer[qear]
      {\lclose}
      {\lsequent{}{2xy+2y(-x)=0}}
    }
    {\lsequent{x^2+y^2=c^2}{\dbox{\hevolve{\D{x}=y\syssep\D{y}=-x}}{x^2+y^2=c^2}}}
\end{sequentdeduction}
But it is not provable with a differential invariant of the form $p>0$.
An invariant of the form \m{p>0} describes an open set and, thus, cannot be equivalent to the (nontrivial) closed domain where \m{x^2+y^2=c^2}.
The only sets that are both open and closed in $\reals^n$ are $\emptyset$ and $\reals^n$.

We extend this argument to arbitrary proofs with an arbitrary number of uses of \irref{diffind}.
By \rref{lem:DIbase}, every \DI[>] proof of the above property uses a non-zero finite number of differential invariants $p_i>0$, such that 
\m{\entails x^2+y^2=c^2 \limply \lorfold_i p_i>0} and \m{\entails p_i>0 \limply x^2+y^2=c^2} for each $i$.
This is a contradiction, because the latter implies that each open \m{p_i>0} is contained in the closed \m{x^2+y^2=c^2}, hence their finite union, which is open, is contained in the closed \m{x^2+y^2=c^2}.
Yet, by the former property, the finite union of the open \m{p_i>0}, which is open, also contains the closed \m{x^2+y^2=c^2}.
This makes \m{x^2+y^2=c^2} both open and closed.
The only sets, however, that are both open and closed in $\reals^n$ are $\emptyset$ and $\reals^n$.
\cexcex{
\begin{sequentdeduction}[array]
  \linfer[closedexpand]
    {\linfer[qear]
      {\linfer[diffind]
        {\linfer[qear]
          {\lclose}
          {\lsequent{}{0\leq \subst[2x\D{x}]{\D{x}}{y}+\subst[2y\D{y}]{\D{y}}{-x}\leq0}}
        }
        {\lsequent{(c-\varepsilon)^2<x^2+y^2<(c+\varepsilon)^2 \land 0<\varepsilon<c}{\dbox{\hevolve{\D{x}=y\syssep\D{y}=-x}}{((c-\varepsilon)^2<x^2+y^2<(c+\varepsilon)^2)}}}
      }
      {\lsequent{x^2+y^2=c^2 \land c>0}{\lforall{0{<}\varepsilon{<}c}{\dbox{\hevolve{\D{x}=y\syssep\D{y}=-x}}{((c-\varepsilon)^2<x^2+y^2<(c+\varepsilon)^2)}}}}
    }
    {\lsequent{x^2+y^2=c^2\land c>0}{\dbox{\hevolve{\D{x}=y\syssep\D{y}=-x}}{x^2+y^2=c^2}}}
\end{sequentdeduction}
}%
\qedhere
\end{proof}

Weak inequalities, however, do subsume the deductive power of equational differential invariants.
This is obvious on the algebraic level but we will see that it also does carry over to the differential structure.

\begin{proposition}[Equational definability] \label{prop:eq-geq}
  The deductive power of differential induction with equations is subsumed by the deductive power of differential induction with weak inequalities:
\[ \DIp[=] \leq \DI[\geq] \]
\end{proposition}
\begin{proof}
\renewcommand{\arraystretch}{1.3}%
By \rref{prop:EDP}, we only need to show that \m{\DI[=] \leq \DI[\geq]}.
Let \m{p=0} be an equational differential invariant of a differential equation \m{\hevolvein{\D{x}=\theta}{\ivr}}.
Then we can prove the following:
\begin{sequentdeduction}[array]
  \linfer[diffind]
    {\linfer%
      {\lclose}
      {\lsequent{\ivr}{\subst[(\D{p}=0)]{\D{x}}{\theta}}}
    }
    {\lsequent{p=0}{\dbox{\hevolvein{\D{x}=\theta}{\ivr}}{p=0}}}
\end{sequentdeduction}
Then, the inequality \m{-p^2\geq0}, which is equivalent to \m{p=0} in real arithmetic, also is a differential invariant of the same dynamics by the following formal proof:
\begin{sequentdeduction}[array]
  \linfer[diffind]
    {\linfer%
      {\lclose}
      {\lsequent{\ivr}{\subst[(-2p\D{p}\geq0)]{\D{x}}{\theta}}}
    }
    {\lsequent{-p^2\geq0}{\dbox{\hevolvein{\D{x}=\theta}{\ivr}}{(-p^2\geq0)}}}
\end{sequentdeduction}
The subgoal for the differential induction step is provable:
if we can prove that $\ivr$ implies \m{\subst[(\D{p}=0)]{\D{x}}{\theta}}, then we can also prove that $\ivr$ implies \m{\subst[(-2p\D{p}\geq0)]{\D{x}}{\theta}}, because \m{\subst[(\D{p}=0)]{\D{x}}{\theta}} implies \m{\subst[(-2p\D{p}\geq0)]{\D{x}}{\theta}} in first-order real arithmetic.
\qedhere
\end{proof}
Note that the local state-based view of differential invariants is crucial to make the last proof work.
By \rref{prop:eq-geq}, differential invariant search with weak inequalities can suppress equations.
Note, however, that the polynomial degree increases quadratically with the reduction in \rref{prop:eq-geq}.
In particular, the polynomial degree increases quartically when using the reductions in \rref{prop:EDP} and \rref{prop:eq-geq} one after another to turn propositional equational formulas into single inequalities.
This quartic increase of the polynomial degree is likely a too serious computational burden for practical purposes even if it is a valid reduction in theory.

When using propositional connectives and inequalities, the reduction is less counterproductive for the polynomial degree.
The following result is an immediate corollary to \rref{prop:eq-geq} but of independent interest. We give a direct proof that shows a more natural reduction that does not increase the polynomial degree.
\begin{corollary}[Atomic equational definability]%
 \label{cor:eq-pgeq}
  The deductive power of differential induction with atomic equations is subsumed by the deductive power of differential induction with formulas with weak inequalities.
  \[ \DI[=] \leq \DIp[\geq] \]
\end{corollary}
\proof
Consider an atomic equational differential invariant of a differential equation system \m{\hevolvein{\D{x}=\theta}{\ivr}}.
We can assume this atomic equational differential invariant to be of the form $p=0$.
If $p=0$ is a differential invariant, then we can show that the formula \m{p\geq0\land -p\geq0} also is a differential invariant by the following formal proof:
\renewcommand{\arraystretch}{1.3}%
\begin{sequentdeduction}[array]
  \linfer
    {\linfer[diffind]
      {\linfer
        {\linfer
          {\lclose}
          {\lsequent{\ivr}{\subst[(\D{p}=0)]{\D{x}}{\theta}}}
        }
        {\lsequent{\ivr}{\subst[(\D{p}\geq0\land-\D{p}\geq0)]{\D{x}}{\theta}}}
      }
      {\lsequent{p\geq0\land-p\geq0}{\dbox{\hevolvein{\D{x}=\theta}{\ivr}}{(p\geq0\land-p\geq0)}}}
    }
    {\lsequent{p=0}{\dbox{\hevolvein{\D{x}=\theta}{\ivr}}{p=0}}}
    \rlap{\hbox to 154 pt{\hfill\qEd}}
\end{sequentdeduction}\medskip

\noindent The same natural reduction works to show the inclusion \m{\DIp[=] \leq \DIp[\geq]} without a penalty for the polynomial degree.
Again, the local state-based view of differential invariants is helpful for this proof.

Next we see that, with the notable exception of pure equations (\rref{prop:EDP}), propositional operators (which have been considered in \cite{DBLP:journals/logcom/Platzer10,DBLP:conf/cav/PlatzerC08,DBLP:journals/fmsd/PlatzerC09} and for some cases also in \cite{DBLP:conf/cav/GulwaniT08} but not in \cite{DBLP:journals/fmsd/SankaranarayananSM08,DBLP:conf/hybrid/PrajnaJ04,DBLP:journals/tac/PrajnaJP07}) increase the deductive power.
\begin{theorem}[Atomic incompleteness] \label{thm:atomic}
  The deductive power of differential induction with propositional combinations of inequalities exceeds the deductive power of differential induction with atomic inequalities.
  \begin{align*} 
  \DI[\geq] &< \DIp[\geq] \\
  \DI[>] &< \DIp[>]
  \end{align*}
\end{theorem}
\begin{proof}
Consider any term $a\geq0$ (e.g., 1 or $x^2+1$ or $x^2+x^4+1$ or $(x-y)^2+2$).
Then the formula \m{x\geq0\land y\geq0\limply\dbox{\hevolve{\D{x}=a\syssep\D{y}=y^2}}{(x\geq0\land y\geq0)}} is provable using a conjunction in the differential invariant:
\renewcommand{\arraystretch}{1.5}%
\begin{sequentdeduction}[array]
  \linfer[diffind]
    {\linfer%
      {\linfer[qear]
        {\lclose}
        {\lsequent{}{a\geq0 \land y^2\geq0}}
      }
      {\lsequent{}{\subst[(\D{x}\geq0\land\D{y}\geq0)]{\D{x}}{a}\subst[\,]{\D{y}}{y^2}}}
    }
    {\lsequent{x\geq0\land y\geq0}{\dbox{\hevolve{\D{x}=a\syssep\D{y}=y^2}}{(x\geq0\land y\geq0)}}}
\end{sequentdeduction}  
By a sign argument similar to that in the proof of \cite[Theorem 2]{DBLP:journals/logcom/Platzer10} no atomic formula is equivalent to \m{x\geq0\land y\geq0}.
Thus, the above property cannot be proven using a single differential induction.
The proof for a postcondition \m{x>0\land y>0} is similar.

This argument extends to arbitrary \DI[\geq] (or \DI[>]) proofs with an arbitrary number of uses of \irref{diffind}.
By \rref{lem:DIbase}, every \DI[\geq] proof of the above property uses a non-zero finite number of differential invariants $p_i(x,y)\geq0$, such that 
\m{\entails x\geq0\land y\geq0 \limply \lorfold_i p_i(x,y)\geq0} and \m{\entails p_i(x,y)\geq0 \limply x\geq0\land y\geq0} for each $i$.
This sign condition easily leads to a contradiction for \DI[\geq], because, for continuity reasons, some $p_i$ can be shown to be zero at infinitely many points $(x,0)$, thus, is zero for all $(x,0)$, contradicting the second condition.
This proof, however, would not work for \DI[>].
  Instead, consider points of the form \m{(x,\lambda x)} for $\lambda,x\in\reals$.
  By the second condition, \m{p_i(x,\lambda x)\leq0} for each $i$ and all \m{x,\lambda\in\reals} with \m{\lambda<0} or \m{x<0}.
  By the first condition, for each \m{\lambda,x>0}, there is an $i$ such that \m{p_i(x,\lambda x)\geq0}.
  By the second condition, this \m{p_i(x,\lambda x)} cannot be a (non-zero) constant polynomial in $x$.
  Since there are only finitely many $i$ to choose from, the same $i$ has to be chosen for infinitely many \m{\lambda,x>0}.
  By cylindrical algebraic decomposition \cite{DBLP:conf/automata/Collins75}, $\reals^2$ partitions into finitely many cells such that each \m{p_i} has constant signs on these cells; see \cite{DBLP:conf/automata/Collins75,BochnakCR98} for details.
  Thus, there are infinitely many $\lambda>0$ for which the same $p_i$ is chosen for all sufficiently large $x>0$.
  Since there are only finitely many cells, there even is some open interval $(s_1,t_1)$ such that the same $p_i$ can be chosen for all \m{\lambda\in(s_1,t_1)} and all sufficiently large $x>0$.
  Write \m{p_i(x,y) = \sum_{j,k=1}^n a_{j,k}x^jy^k}.
  Then,
  \[
  p_i(x,\lambda x) = \sum_{j,k=1}^n a_{j,k}\lambda^k x^{j+k} = \sum_{l=0}^{2n} \underbrace{\sum_{k=0}^l a_{(l-k),k}\lambda^k}_{c_l(\lambda)} x^l
  \]
  Let $d$ be \m{\max_{\lambda\in\reals}\deg p_i(x,\lambda x)}, i.e., the maximum degree of the polynomial \m{p_i(x,\lambda x)}, as a polynomial in $x$.
  By the above sign conditions, $d$ is odd for the infinitely many $\lambda\in(s_1,t_1)$, yet $d$ is even for infinitely many $\lambda<0$ in another open interval $(s_2,t_2)$.
  Consider any of the infinitely many $\lambda$ for which \m{\deg p_i(x,\lambda x)<d}, i.e., \m{c_d(\lambda)=0}.
  Then, the univariate polynomial \m{c_d(\lambda)} has infinitely many zeros, hence, must be the zero polynomial.
  Since $d$ was the degree of \m{p_i(x,\lambda x)}, this means that \m{p_i(x,\lambda x)} is the zero polynomial, which we know for infinitely many $\lambda$ in either the open interval $(s_1,t_1)$ or the interval $(s_2,t_2)$.
  This gives a set of non-zero measure on which $p_i$ is zero, implying that the polynomial $p_i(x,y)$ is the zero polynomial, hence trivial, which is a contradiction.
  \qedhere
\end{proof}
\noindent
Note that the formula in the proof of \rref{thm:atomic} is provable, e.g., using differential cuts (\irref{diffcut}) with two atomic differential induction steps, one for \m{x\geq0} and one for \m{y\geq0}.
Yet, a similar argument can be made to show that the deductive power of differential induction with atomic formulas (even when using differential cuts) is strictly less than the deductive power of general differential induction; see previous work \cite[Theorem~2]{DBLP:journals/logcom/Platzer10}.

Next, we show that differential induction with strict inequalities is incomparable with differential induction with weak inequalities.
In particular, strict and weak barrier certificates are incomparable \cite{DBLP:conf/hybrid/PrajnaJ04,DBLP:journals/tac/PrajnaJP07}.
\begin{proposition}[Elementary incomparability] \label{prop:atom-p}
  The deductive power of differential induction with strict inequalities is incomparable to the deductive power of differential induction with weak inequalities.
  \begin{align*}
   \DI[>] &\not\leq \DIp[\geq] \quad\text{even}~ \DI[>] \not\leq \DIp[\geq,=] \\
   \DI[\geq] &\not\leq \DIp[>] \\
   \DI[=] &\not\leq \DIp[>]
  \end{align*}
\end{proposition}
\begin{proof}
\def\tmpa{a}%
Consider any term $\tmpa>0$ (e.g., $5$ or $x^2+1$ or $x^2+x^4+5$).
The following formula is provable with an atomic differential invariant with a strict inequality:
  \begin{sequentdeduction}[array]
    \linfer[diffind]
    {\linfer[qear]
      {\lclose}
      {\lsequent{}{\tmpa>0}}
    }
    {\lsequent{x>0}{\dbox{\hevolve{\D{x}=\tmpa}}{x>0}}}
  \end{sequentdeduction}
But it is not provable with any conjunctive/disjunctive combination of weak inequalities $p_i\geq0$.
The reason is that the formula \m{x>0} describes a nontrivial open set, which cannot be equivalent to a Boolean formula that is a combination of conjunctions, disjunctions and weak inequalities $p_i\geq0$, because finite unions and intersections of closed sets are closed.
Similarly, the above formula is not provable in \m{\DIp[\geq,=]}, which describe closed regions.

We extend this argument to arbitrary proofs with an arbitrary number of uses of \irref{diffind}.
By \rref{lem:DIbase}, every \DIp[=] or \DIp[\geq,=] proof of the above property uses a non-zero finite number of differential invariants $F_i$ containing only polynomials of the form $p(x)\geq0$ or $p(x)=0$, such that 
\m{\entails O \limply \lorfold_i F_i} and \m{\entails F_i \limply O} for each $i$, where $O$ is the open $x>0$ and each $F_i$ is closed.
This is a contradiction, because the latter implies that each closed \m{F_i} is contained in the open $O$, hence their finite union, which is closed, is contained in the open $O$.
Yet, by the former property, the finite union of the closed $F_i$, which is closed, also contains the open $O$.
This makes $O$ both open and closed, which is impossible in $\reals^n$ except for $\emptyset$ and $\reals^n$.

Conversely, the following formula is provable with an atomic differential invariant with a weak inequality:
  \begin{sequentdeduction}[array]
    \linfer[diffind]
    {\linfer[qear]
      {\lclose}
      {\lsequent{}{\tmpa\geq0}}
    }
    {\lsequent{x\geq0}{\dbox{\hevolve{\D{x}=\tmpa}}{x\geq0}}}
  \end{sequentdeduction}
But it is not provable with any conjunctive/disjunctive combination of strict inequalities $p_i>0$.
The reason is that the formula \m{x\geq0} describes a nontrivial closed set, which cannot be equivalent to a Boolean formula that is a combination of conjunctions, disjunctions and strict inequalities $p_i>0$, because unions and finite intersections of open sets are open.

We extend this argument to arbitrary proofs with an arbitrary number of uses of \irref{diffind}.
By \rref{lem:DIbase}, every \DIp[>] proof of the above property uses a non-zero finite number of differential invariants $F_i$ containing only polynomials of the form $p(x)>0$, such that 
\m{\entails C \limply \lorfold_i F_i} and \m{\entails F_i \limply C} for each $i$, where $C$ is the closed $x\geq0$ and each $F_i$ is open.
This is a contradiction, because the latter implies that each open \m{F_i} is contained in the closed $C$, hence their finite union, which is open, is contained in the closed $C$.
Yet, by the former property, the finite union of the open $F_i$, which is open, also contains the closed $C$.
This makes $C$ both open and closed, which is impossible in $\reals^n$ except for $\emptyset$ and $\reals^n$.

\cexcex{
\begin{sequentdeduction}[array]
  \linfer[closedexpand]
    {\linfer[qear]
      {\linfer[diffind]
        {\linfer[qear]
          {\lclose}
          {\lsequent{}{\tmpa\geq0}}
        }
        {\lsequent{x>-\varepsilon \land \varepsilon>0}{\dbox{\hevolve{\D{x}=\tmpa}}{x>-\varepsilon}}}
      }
      {\lsequent{x\geq0}{\lforall{\varepsilon{>}0}{\dbox{\hevolve{\D{x}=\tmpa}}{x>-\varepsilon}}}}
    }
    {\lsequent{x\geq0}{\dbox{\hevolve{\D{x}=\tmpa}}{x\geq0}}}
\end{sequentdeduction}
}%

Similarly, it is easy to see that \m{\DI[=] \not\leq \DIp[>]}.
By the proof of \rref{prop:gt}, the formula \m{x^2+y^2=c^2\limply\dbox{\hevolve{\D{x}=y\syssep\D{y}=-x}}{x^2+y^2=c^2}} is provable in \m{\DI[=]}.
The formula \m{x^2+y^2=c^2} describes a nontrivial closed set, which, again, cannot be equivalent to any conjunctive/disjunctive combination of strict inequalities $p_i>0$, which would describe an open set.
By \rref{lem:DIbase}, this argument extends to arbitrary proofs in \DIp[>] with an arbitrary number of uses of \irref{diffind} using the same argument as above with $x^2+y^2=c^2$ for $C$.

\cexcex{
\begin{sequentdeduction}[array]
  \linfer[closedexpand]
    {\linfer[qear]
      {\linfer[diffind]
        {\linfer[qear]
          {\lclose}
          {\lsequent{}{0\leq \subst[2x\D{x}]{\D{x}}{y}+\subst[2y\D{y}]{\D{y}}{-x}\leq0}}
        }
        {\lsequent{(c-\varepsilon)^2<x^2+y^2<(c+\varepsilon)^2 \land 0<\varepsilon<c}{\dbox{\hevolve{\D{x}=y\syssep\D{y}=-x}}{((c-\varepsilon)^2<x^2+y^2<(c+\varepsilon)^2)}}}
      }
      {\lsequent{x^2+y^2=c^2 \land c>0}{\lforall{0{<}\varepsilon{<}c}{\dbox{\hevolve{\D{x}=y\syssep\D{y}=-x}}{((c-\varepsilon)^2<x^2+y^2<(c+\varepsilon)^2)}}}}
    }
    {\lsequent{x^2+y^2=c^2\land c>0}{\dbox{\hevolve{\D{x}=y\syssep\D{y}=-x}}{x^2+y^2=c^2}}}
\end{sequentdeduction}
}%
\qedhere
\end{proof}
\begin{corollary} \label{cor:p-incomp}
  We obtain simple consequences:
  \begin{align*}
   \DIp[\geq,=] &\not\leq \DIp[\geq,>,=] \\
   \DIp[=] &\not\leq \DIp[>] \\
   \DIp[>] &\not\leq \DIp[=]
  \end{align*}
\end{corollary}
\begin{proof}
The property \m{\DIp[\geq,=] \not\leq \DIp[\geq,>,=]}  follows from the proof for \m{\DI[\geq] \not\leq \DIp[>]}, because conjunctive/disjunctive combinations of weak inequalities and equations are closed, but the region where \m{x>0} is open.

The separation of \m{\DIp[=]} and \m{\DIp[>]} is a consequence of the facts \m{\DI[=] \not\leq \DIp[>]} and \m{\DI[>] \not\leq \DIp[\geq]}, because \m{\DI[\geq] \geq \DIp[=]} by \rref{prop:eq-geq} and \m{\DIp[=]} describes closed sets yet \m{\DIp[>]} describes open sets.
Using \rref{lem:DIbase}, the extensions to arbitrary proofs with an arbitrary number of uses of \irref{diffind} are as in the proof of \rref{prop:atom-p}.
\qedhere
\end{proof}

Hence, strict inequalities are a necessary ingredient to retain full deductive power.
The operator basis \m{\{\geq,=,\land,\lor\}} is not sufficient.
What about weak inequalities? Do we need those?
The operator basis \m{\{>,\land,\lor\}} is not sufficient by \rref{prop:atom-p}, but what about \m{\{>,=,\land,\lor\}}?
Algebraically, this would be sufficient, because all semialgebraic sets can be defined with polynomials using the operators \m{\{>,=,\land,\lor\}}.
We show that, nevertheless, differential induction with weak inequalities is not subsumed by differential induction with all other operators.
Weak inequalities are thus an inherent ingredient.
In particular, the subsets of operators that have been considered in related work \cite{DBLP:journals/fmsd/SankaranarayananSM08,DBLP:conf/hybrid/PrajnaJ04,DBLP:journals/tac/PrajnaJP07} are not sufficient, since some formulas can only be proven using differential invariants from the full operator basis.

\begin{theorem}[Necessity of full operator basis] \label{thm:pgteq-di}
  The deductive power of differential induction with propositional combinations of strict inequalities and equations is strictly less than the deductive power of general differential induction.
  \begin{align*}
   \DIp[>,=] &< \DIp[\geq,>,=] \\
   \DI[\geq] &\not\leq \DIp[>,=]
  \end{align*}
\end{theorem}
\begin{proof}
The following simple formula is provable with a weak inequality as a differential invariant:
\begin{sequentdeduction}[array]
  \linfer[diffind]
    {\linfer[qear]
      {\lclose}
      {\lsequent{}{1\geq0}}
    }
    {\lsequent{x\geq0}{\dbox{\hevolve{\D{x}=1}}{x\geq0}}}
\end{sequentdeduction}
Suppose $F$ is a propositional formula of strict inequalities and equations that is a differential invariant proving the above formula.
Then $F$ is equivalent to $x\geq0$, which describes a closed region with a nonempty interior.
Consequently, $F$ must have an atom of the form \m{p>0} (otherwise the region has an empty interior or is trivially true and then useless) and an atom of the form \m{q=0} (otherwise the region is not closed).
We can assume $q$ to have a polynomial of degree $\geq1$ (otherwise the region is trivial or not closed if $F$ only has trivially true equations $0=0$ or trivially false equations like $5=0$).
A necessary condition for $F$ to be a differential invariant of \m{\hevolve{\D{x}=1}} thus is that
\begin{equation}
  \entails \subst[(\D{p}>0 \land \D{q}=0)]{\D{x}}{1}
  \label{eq:pgteq-di-diffind}
\end{equation}
because all atoms need to satisfy the differential invariance condition.
Now, $q$ is of the form \m{\sum_{i=0}^n a_i x^i} for some $n,a_0,\dots,a_n$.
Thus, \m{\D{q} = \sum_{i=1}^n i a_i x^{i-1}\D{x}}
and \m{\subst[\D{q}]{\D{x}}{1} = \sum_{i=1}^n i a_i x^{i-1}}.
Consequently, \rref{eq:pgteq-di-diffind} implies that
\[
\entails \sum_{i=1}^n i a_i x^{i-1} = 0
\]
If this formula is valid (true under all interpretations for~$x$), then we must have \m{n\leq1}. Otherwise if $x$ occurs ($n>1$), the above polynomial would not always evaluate to zero.
Consequently $q$ is of the form \m{a_0+a_1x}.
Hence, \m{\subst[(\D{q})]{\D{x}}{1} = a_1}.
Again the validity \rref{eq:pgteq-di-diffind} implies that $a_1$ must be zero.
This contradicts the fact that $q$ has degree $\geq1$.

We extend this argument to arbitrary proofs with an arbitrary number of uses of \irref{diffind}.
By \rref{lem:DIbase}, every \DIp[>,=] proof of the above property uses a non-zero finite number of (nontrivial) differential invariants $F_i$ with only polynomials of the form $p>0$ or $p=0$, such that 
\m{\entails x\geq0 \limply \lorfold_i F_i} and \m{\entails F_i \limply x\geq0} for each $i$.
At least one $F_i$ has to contain a polynomial $q=0$, otherwise each $F_i$ would be open, hence their union would be open, and, by the former condition, contain the closed $x\geq0$, which, by the latter condition, contains each open $F_i$, hence contains their open union.
The only sets that are both open and closed in $\reals^n$ are $\emptyset$ and $\reals^n$.
Our earlier argument about nonempty interiors does not immediately transfer to this situation, because the interior of a union can be bigger than the union of interiors.
Nevertheless, at least one $F_i$ has to contain a polynomial $p>0$, otherwise each $F_i$ would have measure zero (trivial $F_i$ are useless), yet the region \m{x\geq0} of non-zero measure cannot be contained in the region \m{\lorfold_i F_i}, which, as a finite union of sets with measure zero has measure zero.
A necessary condition for all $F_i$ to be differential invariants, thus, is that
\begin{equation*}
  \entails \subst[(\D{p}>0)]{\D{x}}{1}
  \quad\text{and}\quad
  \entails \subst[(\D{q}=0)]{\D{x}}{1}
\end{equation*}
which implies the condition \rref{eq:pgteq-di-diffind} that leads to a contradiction.
\cexcex{
\def\tmpa{1}
\begin{sequentdeduction}[array]
  \linfer[closedexpand]
    {\linfer[qear]
      {\linfer[diffind]
        {\linfer[qear]
          {\lclose}
          {\lsequent{}{\tmpa\geq0}}
        }
        {\lsequent{x>-\varepsilon \land \varepsilon>0}{\dbox{\hevolve{\D{x}=\tmpa}}{x>-\varepsilon}}}
      }
      {\lsequent{x\geq0}{\lforall{\varepsilon{>}0}{\dbox{\hevolve{\D{x}=\tmpa}}{x>-\varepsilon}}}}
    }
    {\lsequent{x\geq0}{\dbox{\hevolve{\D{x}=\tmpa}}{x\geq0}}}
\end{sequentdeduction}
}%
\qedhere
\end{proof}

This completes the study of the relations of classes of differential invariants that we summarize in \rref{fig:diffind-chart} on p.\,\pageref{fig:diffind-chart}.
The other relations are obvious transitive consequences of the ones summarized in \rref{fig:diffind-chart}.

\section{Auxiliary Differential Variable Power}%
\label{sec:diffaux}

After having studied the relationships of the classes of differential invariants, we now turn to extensions of differential induction.
First, we consider auxiliary differential variables, and show that some properties can only be proven after introducing auxiliary differential variables into the dynamics.
That is, the addition of auxiliary differential variables increase the deductive power.
Similar phenomena also hold for classical discrete systems.
Up to now, it was unknown whether auxiliary variables have an effect on provability for the continuous dynamics of differential equations.
We present the following new proof rule \emph{differential auxiliaries} (\irref{diffaux}) for introducing auxiliary differential variables:
\[
\begin{calculuscollections}{10cm}
\begin{calculus}
\cinferenceRule[diffaux|$DA$]{differential auxiliary variables}
{\linferenceRule[sequent]
  {\lsequent[s]{}{\phi\lbisubjunct\lexists{y}{\psi}}
  &\lsequent{\psi} {\dbox{\hevolvein{\D{x}=\theta\syssep\D{y}=\vartheta}{\ivr}}{\psi}}}
  {\lsequent{\phi} {\dbox{\hevolvein{\D{x}=\theta}{\ivr}}{\phi}}}
}{$y$ new and \m{\hevolve{\D{y}=\vartheta,y(0)=y_0}} has a solution $y:[0,\infty)\to\reals^n$ for each $y_0$}%
\end{calculus}
\end{calculuscollections}
\]
Rule \irref{diffaux} is applicable if $y$ is a new variable and the new differential equation \m{\hevolve{\D{y}=\vartheta}} has global solutions on $\ivr$ (e.g., because term $\vartheta$ satisfies a Lipschitz condition \cite[Proposition 10.VII]{Walter:ODE}, which is definable in first-order real arithmetic and thus decidable).
Without that condition, adding \m{\D{y}=\vartheta} could limit the duration of system evolutions incorrectly.
In fact, it would be sufficient for the domains of definition of the solutions of \m{\hevolve{\D{y}=\vartheta}} to be no shorter than those of $x$.
Soundness is easy to see, because precondition $\phi$ implies $\psi$ for some choice of $y$ (left premise).
Yet, for any $y$, $\psi$ is an invariant of the extended dynamics (right premise).
Thus, $\psi$ always holds after the evolution for some $y$ (its value can be different than in the initial state), which still implies $\phi$ (left premise).
Since $y$ is fresh and its differential equation does not limit the duration of solutions of $x$ on $\ivr$, this implies the conclusion.
Since $y$ is fresh, $y$ does not occur in $\ivr$, and, thus, its solution does not leave $\ivr$, which would incorrectly restrict the duration of the evolution as well.

Intuitively, rule \irref{diffaux} can help proving properties, because it may be easier to characterize how $x$ changes in relation to an auxiliary variable $y$ with a suitable differential equation (\m{\hevolve{\D{y}=\vartheta}}).
Let \DIS be the proof calculus with (unrestricted) differential induction (like \DI) plus differential cuts (rule \irref{diffcut}).
We contrast this proof calculus with the extension by \irref{diffaux}.

\begin{theorem}[Auxiliary differential variable power] \label{thm:diffaux-power}
  The deductive power of \DIS with auxiliary differential variables (rule \irref{diffaux}) exceeds the deductive power of \DIS without auxiliary differential variables.
\end{theorem}
\begin{proof}
We show that the formula
\begin{equation}
  x>0 \limply \dbox{\hevolve{\D{x}=-x}}{x>0}
  \label{eq:diffaux-power}
\end{equation}
is provable in \DIS with auxiliary differential variables (rule \irref{diffaux}), but not provable without using auxiliary differential variables.
We first show that \rref{eq:diffaux-power} is provable with auxiliary differential variables (variables that are added and do not affect other formulas or dynamics) using rule \irref{diffaux} (and \irref{diffind}):
\renewcommand{\arraystretch}{1.6}%
\begin{sequentdeduction}[array]
\linfer[diffaux]
 {\linfer[qear]
   {\lclose}
   {\lsequent{}{x>0 \lbisubjunct \lexists{y}{xy^2=1}}}
!
 \linfer[diffind]
    {\linfer%
      {\linfer[qear]
        {\lclose}
        {\lsequent{}{-xy^2+2xy\frac{y}{2}=0}}
      }
      {\lsequent{}{\subst[(\D{x}y^2+x2y\D{y}=0)]{\D{x}}{-x}\subst[\,]{\D{y}}{\frac{y}{2}}}}
    }
    {\lsequent{xy^2=1}{\dbox{\hevolve{\D{x}=-x\syssep\D{y}=\frac{y}{2}}}{xy^2=1}}}
}
{\lsequent{x>0}{\dbox{\hevolve{\D{x}=-x}}{x>0}}}
\end{sequentdeduction}

In the remainder of the proof, we show that \rref{eq:diffaux-power} is not provable without auxiliary differential variables like $y$.
We suppose there was a proof without \irref{diffaux}, which we assume cannot be made shorter (in the number of proof steps and the size of the formulas involved).
Note that for any non-constant univariate polynomial $p$ in the variable~$x$, the limits at $\pm\infty$ exist and are $\pm\infty$, i.e.
\begin{equation}
\lim_{x\to-\infty} p(x) \in \{-\infty,\infty\}
~~\text{and}~~
\lim_{x\to\infty} p(x) \in \{-\infty,\infty\}
\label{eq:univariante-poly-limits}
\end{equation}
For constant polynomials, the limits at $\pm\infty$ exist, are finite, and identical.

Suppose \rref{eq:diffaux-power} were provable by a differential invariant of the form \m{p(x)>0} for a polynomial $p$ in the only occurring variable~$x$.
Then \m{\entails p(x)>0 \lbisubjunct x>0}.
Hence $p(x)$ is not a constant polynomial 
and \m{p(x)\leq0} holds when \m{x\leq0} and \m{p(x)\geq0} when \m{x\geq0} by continuity.
Thus, from \rref{eq:univariante-poly-limits} we conclude
\[
\lim_{x\to-\infty} p(x) = -\infty
\quad\text{and}\quad
\lim_{x\to\infty} p(x) = \infty
\]
In particular, $p(x)$ has the following property, which is equivalent to $p(x)$ having odd degree:
\begin{equation}
  \lim_{x\to-\infty} p(x) \neq \lim_{x\to\infty} p(x)
  \label{eq:diffaux-power-limineq}
\end{equation}
Consequently, the degree of $p$ is odd and the leading (highest-degree) term is of the form $c_{2n+1}x^{2n+1}$ for an $n\in\naturals$ and a number \m{c_{2n+1} \in \reals\setminus\{0\}}.
Since \m{p(x)>0} was assumed to be a differential invariant of \m{\hevolve{\D{x}=-x}}, the differential invariance condition
\m{\entails \subst[(\D{p}>0)]{\D{x}}{-x}} holds.
Abbreviate the polynomial \m{\subst[\D{p}]{\D{x}}{-x}} by $q(x)$.
The leading term of $\D{p}$ is \m{(2n+1)c_{2n+1}x^{2n}\D{x}}.
Consequently, the leading term of $q(x)$ is \m{-(2n+1)c_{2n+1}x^{2n+1}}, hence of odd degree.
Thus $q(x)$ also has the property \rref{eq:diffaux-power-limineq}, which contradicts the fact that the differential invariance condition \m{\subst[(\D{p}>0)]{\D{x}}{-x}}, i.e., \m{q(x)>0} needs to hold for all $x\in\reals$.

Our proof where we suppose that \rref{eq:diffaux-power} were provable by a differential invariant of the form \m{p(x)\geq0} for a polynomial $p$ in the only occurring variable~$x$, and show that this is impossible, is similar, because $p(x)$ then also enjoys property \rref{eq:diffaux-power-limineq}.
Again, a constant polynomial $p(x)$ does not satisfy the requirement \m{\entails p(x)\geq0 \lbisubjunct x>0}.
In fact, the latter equivalence is already impossible, because \m{p(x)\geq0} is closed but \m{x>0} open and nontrivial.

Suppose \rref{eq:diffaux-power} were provable by a differential invariant of the form \m{p(x)=0} for a polynomial $p$ in the only occurring variable~$x$.
Then \m{p(x)=0} must be a consequence of the precondition \m{x>0}.
Thus, the polynomial $p$ is zero at infinitely many points, which implies that this \emph{univariate} polynomial is the zero polynomial.
But \m{0=0} is trivially true and there would be a shorter proof without this useless invariant.
Consequently no single atomic formula can be a differential invariant proving \rref{eq:diffaux-power}.

The fact that \rref{eq:diffaux-power} is not provable in \DI with an arbitrary number of differential invariants that do not need to be single atomic formulas follows from \rref{lem:DIbase}, which implies that \m{x>0} is equivalent to a disjunction \m{\lorfold_i F_i} for a non-zero finite list of (nontrivial) differential invariants $F_i$, because at least one polynomial $p$ in one subformula of one differential invariant $F_i$ must be of the form \m{p(x)\geq0} or \m{p(x)>0} (otherwise the univariate polynomial equations in $F_i$ only have finitely many solutions or are trivial $\ltrue$) and still distinguish $\infty$ from $-\infty$, i.e., satisfy the condition~\rref{eq:diffaux-power-limineq} that leads to a contradiction for the dynamics.
The arguments in the remainder of this proof apply to the disjunctive case as well.

Without differential cuts and \irref{diffaux}, \rref{eq:diffaux-power} is not provable.
Next, suppose \rref{eq:diffaux-power} was provable by differential cuts subsequently with differential invariants \m{F_1,F_2,\dots,F_n}, where each $F_i$ is a logical formula in the only occurring variable~$x$.
Then
\begin{enumerate}[(1)]
\item \label{case:diffaux-power-pre}
  \m{\entails x>0 \limply F_i} for each $i$ (precondition implies each differential invariant), and
\item \label{case:diffaux-power-post}
  \m{\entails F_1 \land \dots \land F_n \limply x>0} (finally implies postcondition), and
\item the respective differential induction step conditions hold.
\end{enumerate}
We abbreviate the conjunction \m{F_1 \land \dots \land F_i} of the first $i$ invariants by $F_{\leq i}$.
Then conditions~\ref{case:diffaux-power-pre} and~\ref{case:diffaux-power-post} imply \m{\entails F_{\leq n} \lbisubjunct x>0}.

By condition~\ref{case:diffaux-power-post}, the region described by $F_{\leq n}$ does not include $-\infty$ (more precisely, this means \m{-\infty \neq \inf \{x : x \models F_{\leq n}\}}).
Hence, there is a smallest $i$ such that the region described by $F_{\leq i}$ does not include $-\infty$ but $F_{\leq i-1}$ still includes $-\infty$.

Then this $F_i$ must have an atomic subformula that distinguishes $\infty$ from $-\infty$ (otherwise $F_{\leq i}$ would have the same truth values for $\infty$ and $-\infty$, and $F_{\leq i}$ would still include $-\infty$, because, by condition~\ref{case:diffaux-power-pre}, all $F_i$ regions include $\infty$).
This atomic subformula has the form \m{p(x)>0} or \m{p(x)\geq0} or \m{p(x)=0} with a univariate polynomial $p(x)$. 
It is easy to see why all univariate polynomial equations \m{p(x)=0} evaluate to false at both $-\infty$ and $\infty$, because of property \rref{eq:univariante-poly-limits}.
Hence, the atomic subformula has the form \m{p(x)>0} or \m{p(x)\geq0} and the univariate polynomial $p(x)$ has to satisfy property \rref{eq:diffaux-power-limineq}, because \m{p(x)>0} or \m{p(x)\geq0} is assumed to distinguish $-\infty$ and $\infty$.
Since the previous domain $F_{\leq i-1}$ still includes $-\infty$ and $\infty$, the same argument as before leads to a contradiction.
In detail.
By property \rref{eq:diffaux-power-limineq}, $p(x)$ has an odd degree.
Since \m{p(x)\geq0} or \m{p(x)>0} was assumed to satisfy the differential invariance condition for \m{\hevolvein{\D{x}=-c}{F_{\leq i-1}}}, it at least satisfies \m{\subst[(\D{p}\geq0)]{\D{x}}{-x}} on the evolution domain $F_{\leq i-1}$.
Because $p(x)$ has odd degree, $\D{p}$ has even degree and the polynomial \m{\subst[\D{p}]{\D{x}}{-x}}, which we abbreviate by $q(x)$, again has odd degree.
Thus $q(x)$ has the property \rref{eq:diffaux-power-limineq}, which contradicts the fact that the differential invariance condition
\m{\subst[(\D{p}\geq0)]{\D{x}}{-x}}, i.e., \m{q(x)\geq0} needs to hold for all $x$ satisfying $F_{\leq i-1}$, hence, at least for $-\infty$ and $\infty$.
This argument applies generally, including to any proof with an arbitrary number of differential invariants and arbitrary number of differential cuts, because some differential invariant still needs to be the first to distinguish $\infty$ from $-\infty$ (otherwise \m{x>0} would not be implied), which entails the condition \rref{eq:diffaux-power-limineq} that leads to a contradiction for the dynamics in \rref{eq:diffaux-power}.
\qedhere
\end{proof}

Note that the same proof can also be used to show that \m{x>0 \limply \dbox{\hevolve{\D{x}=x}}{x>0}} cannot be proven by differential induction and differential cuts without auxiliary differential variables (similarly for other \m{\hevolve{\D{x}=ax}} with a number \m{a\in\reals\setminus\{0\}}).
It is not a barrier certificate \cite{DBLP:journals/tac/PrajnaJP07} either.
Further, the nontrivial open region \m{x>0} cannot be equivalent to the closed region of a barrier certificate \m{p\leq0}.
We decided not to use formula \m{x>0 \limply \dbox{\hevolve{\D{x}=x}}{x>0}} in the proof of \rref{thm:diffaux-power}, however, because it is still provable with what is called \emph{open} differential induction ($\DIop$), where it is sound to assume the differential invariant in the differential induction step if the differential invariant \m{\inv \mequiv x>0} is open \cite{DBLP:journals/logcom/Platzer10}:

\centerline{
\begin{minipage}{4.2cm}
  \begin{sequentdeduction}[array]
    \linfer[diffindop]
    {\linfer[qear]
      {\lclose}
      {\lsequent{x>0}{\subst[(\D{x}>0)]{\D{x}}{x}}}
    }
    {\lsequent{x>0}{\dbox{\hevolve{\D{x}=x}}{x>0}}}
  \end{sequentdeduction}
\end{minipage}
\qquad\qquad~\text{by}~\quad
\renewcommand{\linferenceRuleNameSeparation}{\hspace{2pt}}%
\begin{calculus}
  \cinferenceRule[diffindop|$DI^\circ$]{open differential invariant}
  {\linferenceRule[sequent]
    {\lsequent{\ivr\land\inv}{\subst[\D{\inv}]{\D{x}}{\theta}}}
    {\lsequent{\inv}{\dbox{\hevolvein{\D{x}=\theta}{\ivr}}{\inv}}}
  }{}
\end{calculus}
\quad\text{where~$\inv$ is open}
}%

\noindent
But as an additional result, we show that, because \rref{eq:diffaux-power} has a different sign in the differential equation, also open differential induction is still insufficient for proving \rref{eq:diffaux-power} without the help of auxiliary differential variables.
In particular, our approach can prove a property that related approaches \cite{DBLP:conf/hybrid/PrajnaJ04,DBLP:journals/tac/PrajnaJP07,DBLP:journals/fmsd/SankaranarayananSM08} cannot.

Let $\DISop$ be the calculus with open differential induction (\irref{diffindop}) and differential cuts (\irref{diffcut}).
\begin{theorem}[Open auxiliary differential variable power] \label{thm:diffindop-aux-power}
  The deductive power of \DISop with auxiliary differential variables (rule \irref{diffaux}) exceeds the deductive power of \DISop without auxiliary differential variables.
\end{theorem}
\begin{proof}
In the proof of \rref{thm:diffaux-power} we have shown a formal proof of \rref{eq:diffaux-power} that uses only auxiliary differential variables (\irref{diffaux}) and even only uses regular differential induction (\irref{diffind}) without differential cuts.

In order to see why \rref{eq:diffaux-power} cannot be proven with regular differential induction, open differential induction, and differential cuts without the help of auxiliary differential variables, we continue the proof of \rref{thm:diffaux-power}.
Again we consider the smallest $F_i$ and an atomic subformula \m{p(x)>0} (or \m{p(x)\geq0}) that distinguishes $-\infty$ and $\infty$ with a univariate polynomial $p(x)$.
The point $\infty$ is in $F_{\leq n}$, so there must be such an atomic subformula that is true at $\infty$ and false at $-\infty$.
Consequently, the leading coefficient of $p(x)$ is positive and $p(x)$ enjoys property \rref{eq:diffaux-power-limineq}.
In open differential induction, the differential invariant $F$ can be assumed in the differential induction step whenever the differential invariant $F$ is open.
Thus, the domain in which the differential induction step needs to hold is no longer $F_{\leq i-1}$ but now restricted to \m{F_{\leq i} \mequiv F_{\leq i-1} \land F_i}.
First note that $F_{\leq i-1}$ includes both $\infty$ and $-\infty$ but $F_i$ (and $F_{\leq i}$) only include $\infty$, not $-\infty$.
Then the rest of the proof of \rref{thm:diffaux-power} does not work, because it assumes both $\infty$ and $-\infty$ to matter in the differential invariance condition.

Yet the leading coefficient $c_{2n+1}$ of $p(x)$ is positive and, by \rref{eq:diffaux-power-limineq}, $p(x)$ is of odd degree.
Abbreviate \m{\subst[\D{p}]{\D{x}}{-x}} again by $q(x)$.
Then $q(x)$ is of odd degree and its leading coefficient is negative, because the leading term of $q(x)$ is \m{(2n+1)c_{2n+1}x^{2n}(-x)} and \m{-(2n+1)c_{2n+1}<0}.
But then for $x\to\infty$ (which is in the domain of $F_{\leq i}$), the differential invariant condition \m{q(x)>0} or \m{q(x)\geq0} evaluates to false, which is a contradiction.
\qedhere
\end{proof}

\section{Differential Cut Power}%
\label{sec:diffcut}

Differential cuts (rule \irref{diffcut} on p.\,\pageref{ir:diffcut}) can be used to first prove a lemma about (invariance of a formula along) a differential equation and then restrict the system dynamics to a corresponding subregion.
They are very useful in practice \cite{DBLP:conf/cav/PlatzerC08,DBLP:journals/fmsd/PlatzerC09,Platzer10} especially for finding proofs.
But in some cases, they are just a shortcut for a more difficult proof with a more difficult differential invariant.
This happens, for instance, in the class of air traffic control properties that we had originally conjectured to crucially require differential cuts four years ago \cite{DBLP:journals/logcom/Platzer10}.
Interestingly, no such single invariant was found by a template search with 252 unknowns \cite{DBLP:conf/hybrid/Sankaranarayanan10}.
But we have now found out that it still exists (omitted for space reasons).
Is this always the case?
Can all uses of differential cuts (\irref{diffcut}) be eliminated and turned into a proof of the same property without using \irref{diffcut}?
Is there a differential cut elimination theorem for differential cuts just like there is Gentzen's cut elimination theorem for standard cuts \cite{Gentzen35,Gentzen35I}?
Are all properties that are provable using \irref{diffcut} also provable without \irref{diffcut}?

As the major result of this work, we refute the differential cut elimination hypothesis.
Differential cuts (rule \irref{diffcut}) are \emph{not} just admissible proof rules that can be eliminated, but an inherent proof rule that adds to the deductive power of the proof system.
The addition of differential cuts to differential induction is a significant extension of the deductive power, because, when disallowing differential cuts (like all other approaches do), the deductive power of the proof system strictly decreases.

\begin{theorem}[Differential cut power] \label{thm:diffcut-power}
  The deductive power of differential induction with differential cuts (rule \irref{diffcut}) exceeds the deductive power without differential cuts.
  \[ \DIS > \DI \]
\end{theorem}
The first key insight in the proof of \rref{thm:diffcut-power} is that, for sufficiently large, but fixed, \m{y\gg0} or sufficiently small, but fixed, \m{y\ll0}, the sign of a polynomial \m{p=\sum_{i,j} a_{i,j}x^iy^j} in the limit where either \m{x\to\infty} or \m{x\to-\infty} is determined entirely by the sign of the leading monomial \m{a_{n,m}x^ny^m} with respect to the lexicographical order induced by \m{x\succ y}.
That is, the biggest $n,m \in \naturals$ with \m{a_{n,m}\neq0} such that there is no \m{N>n} and no $j\in\naturals$ with \m{a_{N,j}\neq0} and there is no \m{M>m} with \m{a_{n,M}\neq0}.
The reason why the leading monomial \m{a_{n,m}x^ny^m} dominates is that, for \m{x\to\pm\infty}, the highest degree terms in variable $x$ dominate smaller degree monomials.
Furthermore, for sufficiently large \m{y\gg0} (and for sufficiently small \m{y\ll0}), the highest degree term in variable $y$ among those highest degree terms in $x$ dominates the impact of coefficients of smaller degree.
\begin{proof}[Proof of \rref{thm:diffcut-power}]
\renewcommand{\arraystretch}{1.2}%
Consider the formula
\begin{equation}
  x\geq0\land y\geq0 \limply \dbox{\hevolve{\D{x}=y\syssep\D{y}=1}}{(x\geq0\land y\geq0)}
  \label{eq:diffcut-power}
\end{equation}
First, we show that formula \rref{eq:diffcut-power} is provable easily with a differential cut; see \rref{fig:diffcut-power}.
\begin{figure*}[tb]
\advance\leftskip-1cm
\begin{minipage}{\textwidth}
\small
\begin{sequentdeduction}[array]
  \linfer[diffcut]
    {\linfer[diffind]
      {\linfer%
        {\linfer[qear]
          {\lclose}
          {\lsequent{}{1\geq0}}
        }
        {\lsequent{}{\subst[(\D{y}\geq0})]{\D{x}}{y}\subst[\,]{\D{y}}{1}}
      }
      {\lsequent{x\geq0\land y\geq0}{\dbox{\hevolve{\D{x}=y\syssep\D{y}=1}}{y\geq0}
      }}
    !
   \linfer[diffind]
     {\linfer%
       {\linfer[qear]
         {\lclose}
         {\lsequent{y\geq0}{y\geq0\land 1\geq0}}
       }
       {\lsequent{y\geq0}{\subst[(\D{x}\geq0\land \D{y}\geq0)]{\D{x}}{y}\subst[\,]{\D{y}}{1}}}
     }
    {\lsequent{x\geq0\land y\geq0}{\dbox{\hevolvein{\D{x}=y\syssep\D{y}=1}{y\geq0}}{(x\geq0\land y\geq0)}}}
    }
    {\lsequent{x\geq0\land y\geq0}{\dbox{\hevolve{\D{x}=y\syssep\D{y}=1}}{(x\geq0\land y\geq0)}}}
\end{sequentdeduction}
\end{minipage}
  \caption{Differential cut power: a proof of a simple property that requires differential cuts, not just differential invariants}
  \label{fig:diffcut-power}
\end{figure*}

Now, we need to show that \rref{eq:diffcut-power} is not provable without differential cuts.
Suppose \rref{eq:diffcut-power} was provable by a single differential induction step with a formula $F$ as differential invariant.
Then
\begin{enumerate}[(1)]
\item \label{case:diffcut-power-pre}
  \m{\entails x\geq0 \land y\geq0 \limply F} (precondition implies differential invariant), and
\item \label{case:diffcut-power-post}
  \m{\entails F \limply x\geq0 \land y\geq0} (differential invariant implies postcondition), and
\item \label{case:diffcut-power-step} 
  \m{\entails \subst[\D{F}]{\D{x}}{y}\subst[\,]{\D{y}}{1}} (differential induction step).
\end{enumerate}
By condition~\ref{case:diffcut-power-post}, there has to be a subformula of $F$ in which $x$ occurs (with non-zero coefficient).
This subformula is of the form \m{p\geq0} (or \m{p>0} or \m{p=0})
with a polynomial \m{p:=\sum_{i,j} a_{i,j} x^iy^j}.
By condition~\ref{case:diffcut-power-pre}, there even has to be such a formula of the form \m{p\geq0} or \m{p>0}, because the set described by \m{p=0} has (Lebesgue) measure zero (as $p$ is not the zero polynomial), yet the precondition has non-zero measure (otherwise, if $F$ only had equational subformulas, then the region described by $F$ would have measure zero, contradicting condition~\ref{case:diffcut-power-pre}, or would be trivial \m{0=0}, contradicting condition~\ref{case:diffcut-power-post}).

Consider the leading term \m{a_{n,m}x^ny^m} of $p$ with respect to the lexicographical order induced by \m{x\succ y}.
By condition~\ref{case:diffcut-power-post}, $F$ needs to have a subformula (\m{p\geq0} or \m{p>0}), in which the leading term \m{a_{n,m}x^ny^m} with respect to \m{x\succ y} has odd degree $n$ in $x$ (otherwise, if all leading terms had even degree in $x$, then, for sufficiently large \m{y\gg0}, the truth-values for \m{x\to-\infty} and for \m{x\to\infty} would be identical and, thus, $F$ cannot entail \m{x\geq0} as required by condition~\ref{case:diffcut-power-post}).
By condition~\ref{case:diffcut-power-step}, we know, in particular, that the following holds:
\begin{equation}
  \entails \subst[\D{p}]{\D{x}}{y}\subst[\,]{\D{y}}{1}\geq0
  \quad (\text{or}~ \entails \subst[\D{p}]{\D{x}}{y}\subst[\,]{\D{y}}{1}>0
  ~\text{respectively})
  \label{eq:diffcut-power-step-part} 
\end{equation}
Note that, when forming $\D{F}$ and transforming $p$ into \m{\subst[\D{p}]{\D{x}}{y}\subst[\,]{\D{y}}{1}}, the lexicographical monomial order induced by \m{x\succ y} strictly decreases.
The leading term (with respect to the lexicographical order induced by $x\succ y$) of \m{\subst[\D{p}]{\D{x}}{y}\subst[\,]{\D{y}}{1}} comes from the leading term \m{a_{n,m}x^ny^m} of $p$, and is identical to the leading term of
\begin{align*}
\ell ~:=~& \subst[(na_{n,m}x^{n-1}\D{x}y^m + ma_{n,m}x^ny^{m-1}\D{y})]{\D{x}}{y}\subst[\,]{\D{y}}{1}
\\=~& na_{n,m}x^{n-1}y^{m+1} + ma_{n,m}x^ny^{m-1}
\end{align*}

Now, for sufficiently large \m{y\gg0} or sufficiently small \m{y\ll0}, we see that, in the limit of \m{x\to\pm\infty}, the sign of \m{\subst[\D{p}]{\D{x}}{y}\subst[\,]{\D{y}}{1}} is identical to the sign of $\ell$, because \m{a_{n,m}x^ny^m} is the leading term for the lexicographical order with \m{x\succ y} and the forming of $\D{F}$ does not increase the degree of $x$.
There are two cases to consider:
\begin{iteMize}{$\bullet$}
\item Case $m=0$:
  Then \m{\ell = na_{n,0}x^{n-1}y}.
  Because \rref{eq:diffcut-power-step-part} holds (for all $x,y$), we have, in particular, that
  \begin{enumerate}[(1)]
  \item \m{\ell\geq0} for \m{y\gg0, x\to\pm\infty}. Hence, \m{n-1} is even and \m{a_{n,0}\geq0}.
  \item \m{\ell\geq0} for \m{y\ll0, x\to\pm\infty}. Hence, \m{n-1} is even and \m{a_{n,0}\leq0}.
  \end{enumerate}
  Together, these imply \m{a_{n,0}=0}, which contradicts the fact that \m{a_{n,m}\neq0}, because \m{a_{n,m}} is the leading term.
\item Case $m\neq0$:
  Because \rref{eq:diffcut-power-step-part} holds (for all $x,y$), we have, in particular, that
  \begin{enumerate}[(1)]
  \item \m{\ell\geq0} for \m{y\gg0, x\to\pm\infty}. Then $\ell$ is dominated by the right term \m{ma_{n,m}x^ny^{m-1}}, which has higher degree in $x$. Hence, \m{n} is even and \m{a_{n,m}\geq0}.
  But this contradicts the fact that $n$ is odd.
  \end{enumerate}
\end{iteMize}
In both cases, we have a contradiction, showing that \rref{eq:diffcut-power} is not provable with a differential invariant without differential cuts (\irref{diffcut}).

Finally, we extend the argument to arbitrary proofs with an arbitrary number of uses of \irref{diffind} but without differential cuts (\irref{diffcut}).
By \rref{lem:DIbase}, every \DI proof of \rref{eq:diffcut-power} uses a non-zero finite number of (nontrivial) differential invariants $F_i$, such that each $F_i$ satisfies conditions~\ref{case:diffcut-power-post} and~\ref{case:diffcut-power-step} and
\m{\entails x\geq0\land y\geq0 \limply \lorfold_i F_i}.
Since finite unions of sets of measure zero have measure zero, the latter implies that at least one of the $F_i$ has non-zero measure, which is the only consequence of condition~\ref{case:diffcut-power-pre} that we have used in the first part of the proof.
This formula $F_i$ of non-zero measure, thus, satisfies all the conditions about $F$ that lead to a contradiction (having non-zero measure and conditions~\ref{case:diffcut-power-post} and~\ref{case:diffcut-power-step}).
\qedhere
\end{proof}
For traceability purposes, we use a very simple dynamics in this proof. This particular example could, in fact, still easily be solved with polynomial solutions using auxiliary differential variables (\irref{diffaux}) instead.
Yet, a similar example with more involved dynamics is, e.g., the following, which does not even have a polynomial solution, but is still easily provable by the differential cut \m{y\geq0}:
\[
x\geq0\land y\geq0 \limply \dbox{\hevolve{\D{x}=y\syssep\D{y}=y^4}}{(x\geq0\land y\geq0)}
\]
\section{Related Work}

There are numerous approaches to verifying hybrid systems
\cite{DBLP:conf/lics/Henzinger96,DBLP:conf/hybrid/GreenstreetM99,DBLP:conf/hybrid/AsarinDG03,DBLP:journals/tac/GirardP07,DBLP:journals/sttt/Frehse08}.
We focus our discussion on approaches that are based on proof certificates or similar indirect witnesses for verification.

Approaches based on Lyapunov functions and tangent cones have been used in control, including positively invariant sets and viability theory; see \cite{DBLP:journals/automatica/Blanchini99} for an overview.
These approaches are very successful for linear systems.
Even though the overall theory is interesting, it is purely semantical and defined in terms of limit properties of general functions, which are not computable, even in rich computation frameworks \cite{DBLP:journals/mst/Collins07}.
Similarly, working with solutions of differential equations, which are defined in terms of limits of functions, leads to sound and mathematically well-defined but generally incomputable approaches (except for simple cases like nilpotent linear systems).

The whole point of our approach is that differential invariants are defined in terms of logic and differential algebra and allow us to replace incomputable semantic limit processes by decidable proof rules.
The simplicity of our differential invariants also makes them computationally attractive.
The purpose of this paper is to study the proof theory of differential equations and differential invariants, not just the semantics or mathematical limit processes, which would require arbitrary higher-order logic to reason about.

Differential invariants are related to several other interesting approaches using variations of Lie derivatives, including barrier certificates \cite{DBLP:conf/hybrid/PrajnaJ04,DBLP:journals/tac/PrajnaJP07}, equational templates \cite{DBLP:journals/fmsd/SankaranarayananSM08}, and a constraint-based template approach \cite{DBLP:conf/cav/GulwaniT08}.
Differential invariants provide a general invariance principle of general logical formulas for differential equations.
Other Lie-type approaches to proving invariance \cite{DBLP:conf/hybrid/PrajnaJ04,DBLP:journals/tac/PrajnaJP07,DBLP:journals/fmsd/SankaranarayananSM08,DBLP:conf/cav/GulwaniT08} are either included as a special case of differential invariants or have been suggested with extra assumptions as in the rule \irref{diffindunsound} discussed on p.~\pageref{ir:diffindunsound}, which we do not pursue.

Verification with barrier certificates \cite{DBLP:conf/hybrid/PrajnaJ04} fits to the general rule schema \irref{diffind} where $\inv$ has the special form \m{p\leq0} for a polynomial $p$.
Barrier certificates have also been proposed \cite{DBLP:conf/hybrid/PrajnaJ04} with an extra assumption \m{p=0} in the antecedent of the premise of \irref{diffind}, but that needs a modification in the induction step for soundness \cite{DBLP:journals/tac/PrajnaJP07}.
This modification does not work for more general logical formulas, e.g., with equations.
Even stronger extra assumptions than in \cite{DBLP:conf/hybrid/PrajnaJ04} have been proposed in \cite{DBLP:conf/cav/GulwaniT08}. Modifications of those rules for some special cases have been proposed later on \cite{DBLP:conf/fsttcs/TalyT09}, again based on semantic notions like the tangent cone \cite{DBLP:journals/automatica/Blanchini99}.
Equational templates \cite{DBLP:journals/fmsd/SankaranarayananSM08} are equational differential invariants of the form \m{p=0} for a polynomial $p$, yet with a slightly modified extra assumption.
They do not support inequalities.

Like our first work on differential invariants \cite{DBLP:journals/logcom/Platzer10} and our subsequent search procedure \cite{DBLP:conf/cav/PlatzerC08,DBLP:journals/fmsd/PlatzerC09}, other approaches \cite{DBLP:conf/hybrid/PrajnaJ04,DBLP:journals/tac/PrajnaJP07,DBLP:journals/fmsd/SankaranarayananSM08,DBLP:conf/cav/GulwaniT08} also need to assume that the user provides the right template (sometimes indirectly via a degree bound) to instantiate, but it is not clear how that has to be chosen.
In this paper, we answer the orthogonal question about provability trade-offs in classes of templates, which can better inform decisions about which template classes to consider, for our approach and for others \cite{DBLP:conf/hybrid/PrajnaJ04,DBLP:journals/tac/PrajnaJP07,DBLP:journals/fmsd/SankaranarayananSM08,DBLP:conf/cav/GulwaniT08}.

We consider differential invariants, since they are a sound verification technique, and their logical setting makes it possible to study relative deductive power.
Differential invariants also work for more general logical formulas, which leads to a more interesting and more comprehensive hierarchy of classes.
We have shown how the deductive power increases when considering more general logical formulas as differential invariants than, e.g., the single polynomials considered in other approaches \cite{DBLP:conf/hybrid/PrajnaJ04,DBLP:journals/tac/PrajnaJP07,DBLP:journals/fmsd/SankaranarayananSM08}.
Our results also show that differential cuts and auxiliary differential variables, both of which we have proven to be fundamental proof principles that increase the deductive power, would be interesting additions to other approaches \cite{DBLP:conf/hybrid/PrajnaJ04,DBLP:journals/tac/PrajnaJP07,DBLP:journals/fmsd/SankaranarayananSM08}.
\section{Conclusions and Future Work}

We have considered the differential invariance problem, which, by our relative completeness argument, is at the heart of hybrid systems verification.
To better understand structural properties of hybrid systems, we have identified and analyzed more than a dozen (16) relations between the deductive power of several (9) classes of differential invariants, including subclasses that correspond to related approaches.
Most crucially and surprisingly, we have refuted the differential cut elimination hypothesis and have shown that, unlike standard cuts, differential cuts increase the deductive power.
Our answer to the differential cut elimination hypothesis is the central result of this work.
We have also shown that auxiliary differential variables further increase the deductive power, even in the presence of arbitrary differential cuts.
These findings shed light on fundamental provability properties of hybrid systems and differential equations and are practically important for successful proof search.

Our results require a symbiosis of elements of logic with real arithmetical, differential, semialgebraic, and geometrical properties.
Future work includes investigating this new field further that we call \emph{real differential semialgebraic geometry}, whose development we have only just begun.
One interesting question, e.g., is how our differential invariance chart changes in the presence of disturbances, which differential invariants can handle \cite{DBLP:journals/logcom/Platzer10}.

\section*{Acknowledgement}
I am grateful to the anonymous referees for their valuable feedback and to Frank Herrlich who taught me algebraic geometry.

\bibliographystyle{alpha}
\bibliography{diffcut}

\appendix
\section{Background Proof Rules} \label{app:proof-rules}
Figure~\ref{fig:rules} shows the proof rules that we assume as background rules for our purposes.
They consist of the standard propositional sequent proof rules including the axiom (\irref{axiom}) and \irref{cut} rule for glueing proofs together.
Rule \irref{qear} can use any valid instance of a first-order real arithmetic tautology as a proof rule.
This rule is a simplification of more constructive deduction modulo proof rules for real arithmetic and modular quantifier elimination \cite{DBLP:journals/jar/Platzer08,DBLP:journals/logcom/Platzer10,Platzer10}, which we do not need to consider in this paper.
Except for rule \irref{qear}, the rules in \rref{fig:rules} are standard and listed just for the sake of a complete presentation.
\begin{figure}[htbp]
  \def\globalrule{g}%
  \tabcolsep=1pt%
  \def\leftrule{l}
  \def\rightrule{r}
  \newdimen\linferenceRulehskipamount
  \linferenceRulevskipamount=5mm
    \linferenceRulehskipamount=4mm%
  \newcommand{\dupdate}[2]{\dmodality{\pupdate{#1}}{#2}}%
  \newcommand{\dupdatevar}[1]{\dmodality{\pupdatevar{#1}}}%
  \begin{calculuscollections}{\textwidth}
  \begin{calculuscollection}[prefix=P,reset,context=L]
    \begin{calculus}
      \cinferenceRule[notr|$\lnot$\rightrule]{$\lnot$ right}
      {\linferenceRule[sequent]
        {\lsequent{\phi}{}}
        {\lsequent{}{\lnot \phi}}
      }{}
      \cinferenceRule[notl|$\lnot$\leftrule]{$\lnot$ left}
      {\linferenceRule[sequent]
        {\lsequent{}{\phi}}
        {\lsequent{\lnot \phi}{}}
      }{}
    \end{calculus}
    \hspace{\linferenceRulehskipamount}
    \begin{calculus}
      \cinferenceRule[orr|$\lor$\rightrule]{$\lor$ right}
      {\linferenceRule[sequent]
        {\lsequent{}{\phi, \psi}}
        {\lsequent{}{\phi \lor \psi}}
      }{}
      \cinferenceRule[orl|$\lor$\leftrule]{$\lor$ left}
      {\linferenceRule[sequent]
        {\lsequent{\phi}{}
          & \lsequent{\psi}{}}
        {\lsequent{\phi \lor \psi}{}}
      }{}
    \end{calculus}
    \hspace{\linferenceRulehskipamount}
    \begin{calculus}
      \cinferenceRule[andr|$\land$\rightrule]{$\land$ right}
      {\linferenceRule[sequent]
        {\lsequent{}{\phi}
          & \lsequent{}{\psi}}
        {\lsequent{}{\phi \land \psi}}
      }{}
      \cinferenceRule[andl|$\land$\leftrule]{$\land$ left}
      {\linferenceRule[sequent]
        {\lsequent{\phi , \psi}{}}
        {\lsequent{\phi \land \psi}{}}
      }{}
    \end{calculus}
    \hspace{\linferenceRulehskipamount}
    \begin{calculus}
      \cinferenceRule[implyr|$\limply$\rightrule]{$\limply$ right}
      {\linferenceRule[sequent]
        {\lsequent{\phi}{\psi}}
        {\lsequent{}{\phi \limply \psi}}
      }{}
      \cinferenceRule[implyl|$\limply$\leftrule]{$\limply$ left}
      {\linferenceRule[sequent]
        {\lsequent{}{\phi}
          & \lsequent{\psi}{}}
        {\lsequent{\phi \limply \psi}{}}
      }{}
    \end{calculus}
    \hspace{\linferenceRulehskipamount}
    \begin{calculus}
      \cinferenceRule[axiom|$ax$]{axiom}
      {\linferenceRule[sequent]
        {}
        {\lsequent{\phi}{\phi}}
      }{}
      \cinferenceRule[cut|$cut$]{cut}
      {\linferenceRule[sequent]
        {\lsequent{}{\phi}
        &\lsequent{\phi}{}}
        {\lsequent{}{}}
      }{}
    \end{calculus}
    \hspace{\linferenceRulehskipamount}\hspace{1.9cm}
    \begin{calculus}
      \cinferenceRule[qear|\usebox{\Rval}]{real arithmetic}
      {\linferenceRule[sequent]
        {\lsequent[l]{\tilde{\Gamma}}{\tilde{\Delta}}}
        {\lsequent[l]{\Gamma}{\Delta}}
      }{if $(\tilde{\Gamma}\limply\tilde{\Delta})\limply(\Gamma\limply\Delta)$ is an instance of a valid tautology of first-order real arithmetic}
    \end{calculus}
  \end{calculuscollection}
  \end{calculuscollections}
  \caption{Basic proof rules}
  \label{fig:rules}
\end{figure}

Despite the presence of arbitrary cuts (\irref{cut}) in our proof calculus, only differential invariants (using rule \irref{diffind}) and, for simple properties, differential weakening (rule \irref{diffweak}) ultimately prove properties of differential equations.
The propositional proof rules (including \irref{cut}) and real arithmetic rule (\irref{qear}) can perform arbitrary case distinctions, consider subregions, transform and handle the arithmetic etc., all of which is necessary.
But they cannot, on their own, prove properties of differential equations that are not just mere first-order tautologies.
Likewise, differential cuts (rule \irref{diffcut}) are very powerful, but still ultimately need differential invariants (rule \irref{diffind}) to prove at least their respective left premise.
By inspecting the proof rules, it is easy to see that the only proof rules in \DI that can derive nontrivial properties of differential equations are \irref{diffind}, \irref{diffweak}, and \irref{genb}. 
Rule \irref{diffcut} is not part of the \DI calculus and, thus, rule \irref{diffweak} cannot derive a formula without evolution domain region.
Rule \irref{genb} cannot be used instead of \irref{diffind}, only in addition to \irref{diffind}, because the extra premise resulting from \irref{genb} still needs to be proved, ultimately, using \irref{diffind} for a sufficiently strong differential invariant.
We, thus, prove that, if no differential cuts with rule \irref{diffcut} are used, then all provable properties of differential equations follow directly or indirectly from differential invariants (\irref{diffind}).
This result depends on the proof rules we consider.
\begin{lemma}[Differential invariance base] \label{lem:DIbase}
  If \m{\lsequent{A}{\dbox{\hevolve{\D{x}=\theta}}{B}}} is provable in \DI[\Omega] for a set $\Omega$ of operators and nontrivial first-order formulas $A,B$ (i.e., \m{A\mnequiv\lfalse}, \m{B\mnequiv\ltrue}),
  then there is a (non-zero) finite number of \dL formulas \m{\lsequent{F_i}{\dbox{\hevolve{\D{x}=\theta}}{F_i}}} provable by \irref{diffind} in \DI[\Omega] with quantifier-free first-order formulas $F_i$ over the operators $\Omega$ such that the following \dL formulas are valid
  \begin{align*}
    \entails & A\limply\lorfold_i F_i\\
    \entails & \landfold_i (F_i\limply B)
  \end{align*}
\end{lemma}
\begin{proof}
\newcommand*{\hevolveinh}[2]{\hevolve{#1}}%
Consider a proof of goal \m{\lsequent{A}{\dbox{\hevolveinh{\D{x}=\theta}{\ivr}}{B}}} with a minimal number of \irref{cut} uses among the proofs with a minimal number of \irref{genb} uses.
All minimal proofs of \m{\lsequent{A}{\dbox{\hevolveinh{\D{x}=\theta}{\ivr}}{B}}} have the following form (there can be additional premises that are first-order tautologies, close by first-order reasoning, and are of no relevance for this result):
\begin{sequentdeduction}
  \linfer
    {\linfer
      {\linfer{\lclose}{\linfer{\vdots}{\lsequent{\Gamma_1}{\dbox{\hevolvein{\D{x}=\theta}{\ivr_1}}{B_1}}}}
      & \dots
      & \linfer{\lclose}{\linfer{\vdots}{\lsequent{\Gamma_n}{\dbox{\hevolvein{\D{x}=\theta}{\ivr_n}}{B_n}}}}}
      {\lsequent{}{\vdots}}
    }
  {\lsequent{A}{\dbox{\hevolveinh{\D{x}=\theta}{\ivr}}{B}}}
\end{sequentdeduction}
The \m{\lsequent{\Gamma_i}{\dbox{\hevolvein{\D{x}=\theta}{\ivr_i}}{B_i}}} are proved by one of the global rules \irref{diffind}, \irref{diffweak}, \irref{diffcut}, or \irref{genb}
and all rules that conclude \m{\lsequent{A}{\dbox{\hevolveinh{\D{x}=\theta}{\ivr}}{B}}} from the \m{\lsequent{\Gamma_i}{\dbox{\hevolvein{\D{x}=\theta}{\ivr_i}}{B_i}}} are propositional rules or rule \irref{qear}.
In particular, the following is a first-order tautology:
\begin{equation}
\landfold_{i=1}^n (\lsequent{\Gamma_i}{\dbox{\hevolvein{\D{x}=\theta}{\ivr_i}}{B_i}}) \limply (\lsequent{A}{\dbox{\hevolveinh{\D{x}=\theta}{\ivr}}{B}})
\label{eq:FOLpart}
\end{equation}
Formula \rref{eq:FOLpart} is FOL-valid, its premises \m{\landfold_i (\lsequent{\Gamma_i}{\dbox{\hevolvein{\D{x}=\theta}{\ivr_i}}{B_i}})} are \dL-valid, the antecedent $A$ is \dL-satisfiable and the succedent \m{\dbox{\hevolveinh{\D{x}=\theta}{\ivr}}{B}} is not \dL-valid.
Without loss of generality, we can assume there are no extra succedents $\Delta_i$ and all $\Gamma_i$ have no superfluous formulas by weakening (a special case of rule \irref{qear}) in the global premises.

At least one \m{\lsequent{\Gamma_i}{\dbox{\hevolvein{\D{x}=\theta}{\ivr_i}}{B_i}}} occurs that is proved by a global rule, otherwise the proof would be a first-order tautology, yet, the goal is not a first-order tautology since $A$ is first-order yet its the succedent is nontrivial. The propositional rules and rule \irref{qear} are sound for first-order real arithmetic so they cannot derive formulas that are no first-order tautologies. 
If one of the global branches, say, \m{\lsequent{\Gamma_n}{\dbox{\hevolvein{\D{x}=\theta}{\ivr_n}}{B_n}}} is proved by \irref{genb}, then for some $C_n$ the formula \m{C_n\limply B_n} is valid and $\Gamma_n$ is (by weakening via rule \irref{qear}) identical to \m{\dbox{\hevolvein{\D{x}=\theta}{\ivr_n}}{C_n}}. This formula $\Gamma_n$ also needs to appear positively in another branch,
because it is not in the goal, so it could only have been introduced by a \irref{cut}, which leaves the formula in some antecedent.
Note that weakening (via rule \irref{qear}) could remove the positive formula \m{\dbox{\hevolvein{\D{x}=\theta}{\ivr_n}}{C_n}} again, but not from all branches, since then there would have been a simpler proof with less uses of \irref{cut}, contradicting our minimality assumption.
This modal formula \m{\dbox{\hevolvein{\D{x}=\theta}{\ivr_n}}{C_n}} also 
cannot be provable trivially just by propositional and \irref{qear} rules alone, because it must then have been introduced by a \irref{cut} before and the goal would then be provable with less uses of \irref{cut}.
Hence, \m{\dbox{\hevolvein{\D{x}=\theta}{\ivr_n}}{C_n}} is proved by a global rule and must, thus, be identical to some succedent \m{\dbox{\hevolvein{\D{x}=\theta}{\ivr_i}}{B_i}}.
By the form of \irref{genb}, \m{B_i\limply B_n} is valid.
If \m{\dbox{\hevolvein{\D{x}=\theta}{\ivr_i}}{B_i}} is proved by \irref{genb}, the argument repeats and some global branch has to be proved by \irref{diffind}, otherwise the proof does not close.
Note that \irref{diffweak} would not suffice, because \m{B\mnequiv\ltrue}, yet the evolution domain constraint in the goal is trivial ($\ltrue$) and the only rule that can conclude properties without evolution domain constraints from properties with evolution domain constraints, i.e., rule \irref{diffcut}, is not allowed in \DI[\Omega].

We can assume that \m{\lsequent{B_1}{\dbox{\hevolvein{\D{x}=\theta}{\ivr_1}}{B_1}}}, \dots, \m{\lsequent{B_r}{\dbox{\hevolvein{\D{x}=\theta}{\ivr_r}}{B_r}}} for some $r\geq1$ are exactly the global branches that are proved using \irref{diffind} just by reordering the branches.
By weakening (via rule \irref{qear}), no superfluous formulas occur in the $\Gamma_i$.
Furthermore, $B_i$ is in first-order logic over the operators $\Omega$ by assumption.
Since rule \irref{diffweak} does not help without \irref{diffcut} for proving the goal (and irrelevant uses of \irref{diffweak} violate our minimality assumption), the remaining global rule applications are \irref{genb}.
Overall, the first-order tautology \rref{eq:FOLpart}, thus, is of the following form:
\begin{equation}
\landfold_{i=1}^r (\lsequent{B_i}{\dbox{\hevolvein{\D{x}=\theta}{\ivr_i}}{B_i}}) \land
\landfold_{i=r+1}^n (\lsequent{\dbox{\hevolvein{\D{x}=\theta}{\ivr_i}}{B_{l_i}}}{\dbox{\hevolvein{\D{x}=\theta}{\ivr_i}}{B_i}}) \limply (\lsequent{A}{\dbox{\hevolveinh{\D{x}=\theta}{\ivr}}{B}})
\label{eq:FOLpartb}
\end{equation}
By the above argument about the recursive structure of \irref{genb}, we can further assume the branches are ordered such that $l_i<i$ (no cyclic generalizations, which would be unsound).
Note that formula \rref{eq:FOLpartb} is a Horn-logic query with query \m{\dbox{\hevolveinh{\D{x}=\theta}{\ivr}}{B}}, atomic facts $A$, and the antecedent as Horn rules.

Consider any \dL state $\iportray{\I}$ with \m{\imodels{\I}{A}}, which exists by assumption.
Since the antecedent of \rref{eq:FOLpartb} is valid in \dL, it is also true in $\iportray{\I}$.
Since \rref{eq:FOLpartb} is a first-order tautology, some subformula of \rref{eq:FOLpartb} forces \m{\dbox{\hevolveinh{\D{x}=\theta}{\ivr}}{B}} to be true in the state $\iportray{\I}$ (considered as a state of first-order real arithmetic).
By the shape of \rref{eq:FOLpartb}, this implies that the antecedent of \rref{eq:FOLpartb} has to contain some implication \m{\lsequent{\Gamma_i}{\dbox{\hevolvein{\D{x}=\theta}{\ivr}}{B_i}}} with $B_i$ being $B$ such that \m{\imodels{\I}{\Gamma_i}}; otherwise \m{\lnot\dbox{\hevolvein{\D{x}=\theta}{\ivr}}{B}} would be a first-order consistent extension of state $\iportray{\I}$ that invalidates \rref{eq:FOLpartb}.
If $\Gamma_i$ is in first-order logic (i.e., $i\leq r$), then we are done when choosing $\Gamma_i$ to be among the $F_k$.
Otherwise, $\Gamma_i$ is of the form \m{\dbox{\hevolvein{\D{x}=\theta}{\ivr}}{B_{l_i}}} (recall that \m{\entails B_{l_i}\limply B_i}) and must be forced by another formula \m{\lsequent{\Gamma_j}{\dbox{\hevolvein{\D{x}=\theta}{\ivr}}{B_{l_i}}}} in \rref{eq:FOLpartb} with \m{\imodels{\I}{\Gamma_j}} and so on recursively using the  same argument for $\Gamma_j$ instead of $\Gamma_i$.
Repeating the argument for all \dL states $\iportray{\I}$, we obtain a finite number of formulas $F_k$ with the desired property: each $F_k$ is a differential invariant in \DI[\Omega], implies $B$, and their disjunction holds in any state where $A$ holds.
These are finitely many $F_k$ since the formula \rref{eq:FOLpartb} where the $F_k$ are selected from is finite.
There is at least one $F_k$, because, as above, the goal itself is no first-order tautology.
\qedhere
\end{proof}

A similar result holds in the presence of evolution domain constraints.

\section{Soundness of Differential Induction and Differential Cuts} \label{app:sound}
{\newcommand{\crf}{c}%
\newcommand{\der}[2][]{\subst[#2']{\D{x}}{\theta}}%
\newcommand{\If}{\DALint[flow=\varphi]}%
\providecommand{\interval}[1]{#1}%
\providecommand{\argholder}[0]{\cdot}%
We have proved soundness of proof rules \irref{diffind} and \irref{diffcut} and the other rules in previous work \cite{DBLP:journals/logcom/Platzer10}.
In the interest of a self-contained presentation, we repeat the critical soundness proofs here in a simplified form that directly uses the notation of this paper.

For the proof of soundness of \irref{diffind}, we first prove \rref{lem:derivationLemma}, which says that the valuation of syntactic total derivation~$\der{F}$ (with differential equations substituted in) of formula~$F$ as defined in \rref{sec:diffind} coincides with analytic differentiation.
We first prove this derivation lemma for terms~$c$.
\proof[Proof of \rref{lem:derivationLemma}]
  \newcommand{\Iff}{\iconcat[state=\varphi(t)]{\I}}%
  \newcommand{\Ifz}{\iconcat[state=\varphi(\zeta)]{\I}}%
  The proof is by induction on term~$\crf$.
  The differential equation~$\hevolve{\D{x}=\theta}$ is of the form \m{\hevolve{\D{x_1}=\theta_1\syssep\dots\syssep\D{x_n}=\theta_n}}.
  \newcommand{\Ds}[2][]{\D[#1]{}(#2)}%
  \begin{iteMize}{$\bullet$}
   \item If~$\crf$ is one of the variables~$x_j$ for some~$j$ (for other variables, the proof is simple because~$\crf$ is constant during~$\iget[flow]{\If}$) then:
    \begin{displaymath}
      \D[t]{\,{\ivaluation{\Iff}{x_j}}}(\zeta)
      = \ivaluation{\Ifz}{\theta_j}
      = \ivaluation{\Ifz}{\sum_{i=1}^n\Dp[x_i]{x_j}\theta_i}
      \enspace.
    \end{displaymath}
    The first equation holds by definition of the semantics.
    The last equation holds as \m{\Dp[x_j]{x_j}=1} and \m{\Dp[x_i]{x_j}=0} for~$i\neq j$.
    The derivatives exist because~$\iget[flow]{\If}$ is (continuously) differentiable for $x_j$.
   \item If~$\crf$ is of the form~$a+b$, the desired result can be obtained by using the properties of derivatives and semantic valuation:
    \begin{align*}
      & \phantom{=~}
      \D[t]{\,\ivaluation{\Iff}{a+b}}(\zeta)
      && \\
      & =
      \D[t]{\,(\ivaluation{\Iff}{a} + \ivaluation{\Iff}{b})}(\zeta)
      && \ivaluation{\I}{\argholder}~\text{is a linear operator for all}~\iportray{\I}\\
      & =
      \D[t]{\,\ivaluation{\Iff}{a}}(\zeta) + \D[t]{\,\ivaluation{\Iff}{b}}(\zeta)
      && \D[t]{}~\text{is a linear operator}\\
      & =
      \ivaluation{\Ifz}{\der{a}} + \ivaluation{\Ifz}{\der{b}}
      && \text{by induction hypothesis}\\
      & =
      \ivaluation{\Ifz}{\der{a} + \der{b}}
      && \ivaluation{\I}{\argholder}~\text{is a linear operator for}~\iget[state]{\I}=\iget[state]{\Ifz}\\
      & =
      \ivaluation{\Ifz}{\der{(a + b)}}
      && \text{derivation is linear, because}~\Dp[x_i]{}~\text{is~linear}
    \end{align*}
   \item The case where~$\crf$ is of the form~$a \cdot b$ is accordingly, using Leibniz's product rule for \m{\Dp[x_i]{}}; see \cite{Platzer10}.\qed
  \end{iteMize}\smallskip

\begin{proof}[Proof of Soundness of \irref{diffind}]
  {\newcommand{\iFf}[1][t]{\iconcat[state=\varphi(#1)]{\I}}%
  \newcommand{\iFz}{\iFf[\zeta]}%
  \newcommand{\Ifx}{\iFf[\xi]}%
  In order to prove soundness of rule \irref{diffind}, we need to prove that, whenever the premise is valid (true in all states), then the conclusion is valid.
  We have to show that~\m{\imodels{\I}{\inv\limply\dbox{\hevolvein{\hevolve{\D{x}=\theta}}{\ivr}}{\inv}}} for all states~$\iportray{\I}$.
  Let~$\iportray{\I}$ satisfy~\m{\imodels{\I}{\inv}} as, otherwise, there is nothing to show.
  We can assume~$F$ to be in disjunctive normal form and consider any disjunct~$G$ of~$F$ that is true at~$\iportray{\I}$.
  In order to show that~$F$ remains true during the continuous evolution, it is sufficient to show that each conjunct of~$G$ is.
  We can assume these conjuncts to be of the form~$\crf\geq0$ (or~$\crf>0$ where the proof is accordingly).
  Finally, using vectorial notation, we write~$\D{x}=\theta$ for the differential equation system.
  Now let~\m{\iget[flow]{\If}:\interval{[0,r]}\to(V\to\reals)} be any flow of~\m{\hevolvein{\D{x}=\theta}{\ivr}} beginning in~$\iget[flow]{\If}(0)=\iget[state]{\I}$.
  If the duration of~$\iget[flow]{\If}$ is~$r=0$, we have \m{\imodels{{\iFf[0]}}{\crf\geq0}} immediately, because \m{\imodels{\I}{\crf\geq0}}.
  For duration~$r>0$, we show that~$\crf\geq0$ holds all along the flow~$\iget[flow]{\If}$, i.e.,~\m{\imodels{\iFz}{\crf\geq0}} for all~\m{\zeta\in\interval{[0,r]}}.
  
  Suppose there was a~$\zeta\in\interval{[0,r]}$ with~\m{\imodels{\iFz}{\crf<0}}, which will lead to a contradiction.
  The function~$h:\interval{[0,r]}\to\reals$ defined as~\m{h(t)=\ivaluation{\iFf}{\crf}} satisfies the relation~\m{h(0)\geq0>h(\zeta)}, because \m{h(0)=\ivaluation{\iFf[0]}{\crf}=\ivaluation{\I}{\crf}} and~\m{\imodels{\I}{\crf\geq0}} by antecedent of the conclusion.
  By \rref{lem:derivationLemma}, $h$ is continuous on~$\interval{[0,r]}$ and differentiable at every~$\xi\in\interval{(0,r)}$.
  By mean value theorem,
  there is a~$\xi\in\interval{(0,\zeta)}$ such that
  \m{\D[t]{h(t)}(\xi)\cdot(\zeta-0) = h(\zeta)-h(0)<0}.
  In particular, since~$\zeta\geq0$, we can conclude that~\m{\D[t]{h(t)}(\xi) < 0}.
  Now \rref{lem:derivationLemma} implies that
  \m{\D[t]{h(t)}(\xi) = \ivaluation{\Ifx}{\der{\crf}}<0}.
  This, however, is a contradiction, because the premise implies that the formula \m{\ivr\limply\der{(\crf\geq0)}} is true in all states along~$\iget[flow]{\If}$,
  including \m{\imodels{\Ifx}{\ivr\limply\der{(\crf\geq0)}}}.
  In particular, as~$\iget[flow]{\If}$ is a flow for~\m{\hevolvein{\D{x}=\theta}{\ivr}}, we know that~\m{\imodels{\Ifx}{\ivr}} holds, and we have~\m{\imodels{\Ifx}{\der{(\crf\geq0)}}}, 
  which contradicts \m{\ivaluation{\Ifx}{\der{\crf}}<0}.
  }%
  \qedhere
\end{proof}\smallskip

\begin{proof}[Proof of Soundness of \irref{diffcut}]
  Rule \irref{diffcut} is sound using the fact that the left premise implies that every flow~$\iget[flow]{\If}$ that satisfies~\m{\D{x}=\theta} also satisfies~$\ivr$ \emph{all along} the flow.
  Thus, if flow $\iget[flow]{\If}$ satisfies \m{\D{x}=\theta}, it also satisfies \m{\hevolvein{\D{x}=\theta}{\ivr}},
  so that the right premise entails the conclusion.
\qedhere
\end{proof}
}%

\end{document}